\documentclass[acmsmall,nonacm]{acmart}
\pdfoutput=1
\usepackage{xspace}
\usepackage{graphicx}
\usepackage{subfig}
\usepackage{mathtools}
\usepackage{amsmath}
\usepackage{dsfont}
\usepackage{adjustbox}
\usepackage{tikz}
\usepackage{caption}
\usepackage{enumitem}
\graphicspath{{./figs/}}
\newcommand{\eg}{\emph{e.g.}}
\newcommand{\ie}{\emph{i.e.}}

\newcommand{\cf}{\emph{cf.}\xspace}
\newcommand{\vect}[1]{\boldsymbol{#1}}
\allowdisplaybreaks
\newcommand{\gv}[1]{\textcolor{magenta}{#1}}

\newcommand{\id}{\mathbb{I}}

\newcommand{\proj}[1]{| \hspace{1pt} #1 \rangle \langle #1 \hspace{1pt} |}

\newcommand{\ket}[1]{|#1\rangle}               
\newcommand{\bra}[1]{\langle #1|}              

\newcommand{\Complex}{\mathbb{C}}
\newcommand{\tr}{\mathsf{Tr}}

\newcommand{\Tr}{{\rm Tr}}

\newtheorem{proposition}{Proposition}
\newtheorem{lemma}{Lemma}
\AtBeginDocument{%
  \providecommand\BibTeX{{%
    \normalfont B\kern-0.5em{\scshape i\kern-0.25em b}\kern-0.8em\TeX}}}

\setcopyright{none}
\renewcommand\footnotetextcopyrightpermission[1]{}
\begin{document}

\title{On the Quantum Performance Evaluation of Two Distributed Quantum Architectures}

\author{Gayane Vardoyan}
\email{g.s.vardoyan@tudelft.nl}
\author{Matthew Skrzypczyk}
\email{m.d.skrzypczyk@tudelft.nl}
\author{Stephanie Wehner}
\email{s.d.c.wehner@tudelft.nl}
\affiliation{%
  \institution{QuTech and Kavli Institute of Nanoscience, Delft University of Technology}
  \address{Lorentzweg 1, 2628 CJ Delft}
  \city{Delft}
  \country{The Netherlands}
}
\if{false}
\author{Matthew Skrzypczyk}
\affiliation{%
\email{m.d.skrzypczyk@tudelft.nl}}

\author{Stephanie Wehner}
\affiliation{
\email{s.d.c.wehner@tudelft.nl}
}
\fi

\begin{abstract}
Distributed quantum applications impose requirements on the quality of the quantum states that they consume. When analyzing architecture implementations of quantum hardware, characterizing this quality forms an important factor in understanding their performance.
Fundamental characteristics of quantum hardware lead to inherent tradeoffs between the quality of states and traditional performance metrics such as throughput. Furthermore, any real-world implementation of quantum hardware exhibits time-dependent noise that degrades the quality of quantum states over time.
Here, we study the performance of two possible architectures for interfacing a quantum processor with a quantum network. The first corresponds to the current experimental state of the art in which the same device functions both as a processor and a network device. The second corresponds to a future architecture that separates these two functions over two distinct devices. 
We model these architectures as continuous-time Markov chains and compare their quality of executing quantum operations and producing entangled quantum states as functions of their memory lifetimes, as well as the time that it takes to perform various operations within each architecture. As an illustrative example, we apply our analysis to architectures based on Nitrogen-Vacancy centers in diamond, where we find that for present-day device parameters one architecture is more suited to computation-heavy applications, and the other for network-heavy ones. 
We validate our analysis with the quantum network simulator NetSquid.
Besides the detailed study of these architectures, a novel contribution of our work are several formulas that connect an understanding of waiting time distributions to the decay of quantum quality over time for the most common noise models employed in quantum technologies. This provides a valuable 
new tool for performance evaluation experts, and its applications extend beyond the two architectures studied in this work.
\end{abstract}



\maketitle

\section{Introduction}
\if{false}
\begin{figure}[t]
{\centering
\begin{minipage}{0.445\textwidth}
  \centering
  \subfloat[single-device architecture]{\includegraphics[width=\textwidth]{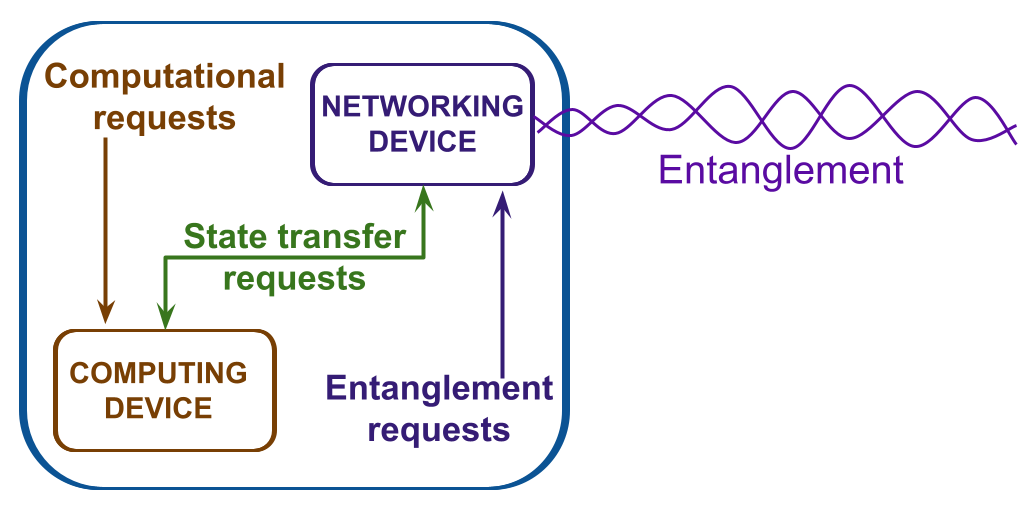}}\\
  \subfloat[double-device architecture]{\includegraphics[width=\textwidth]{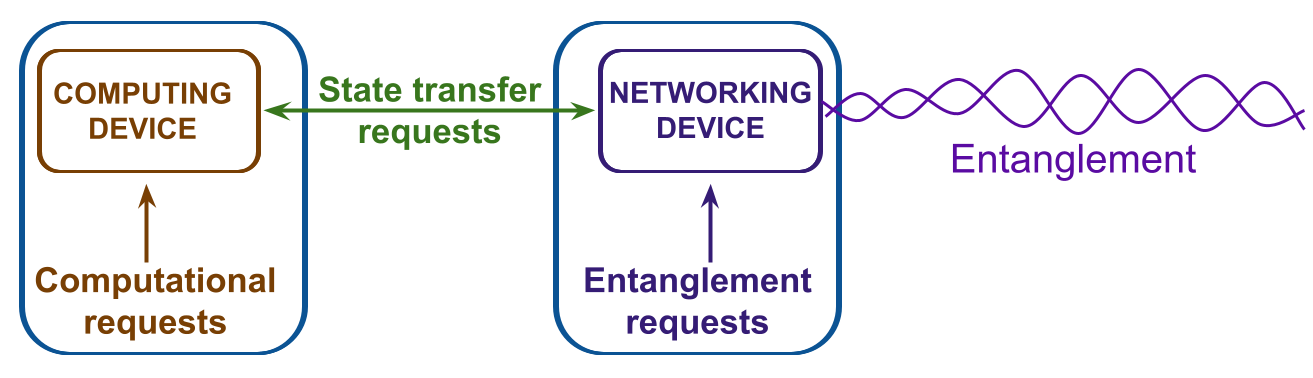}}
  \captionof{figure}{A quantum architecture with computing and networking components. State transfer operations require cooperation from both components.}
  \label{fig:archAbstract}
  \end{minipage}\qquad\quad
  \begin{minipage}{0.445\textwidth}
  \centering
  \subfloat[the single-NV architecture]{\includegraphics[width=\textwidth]{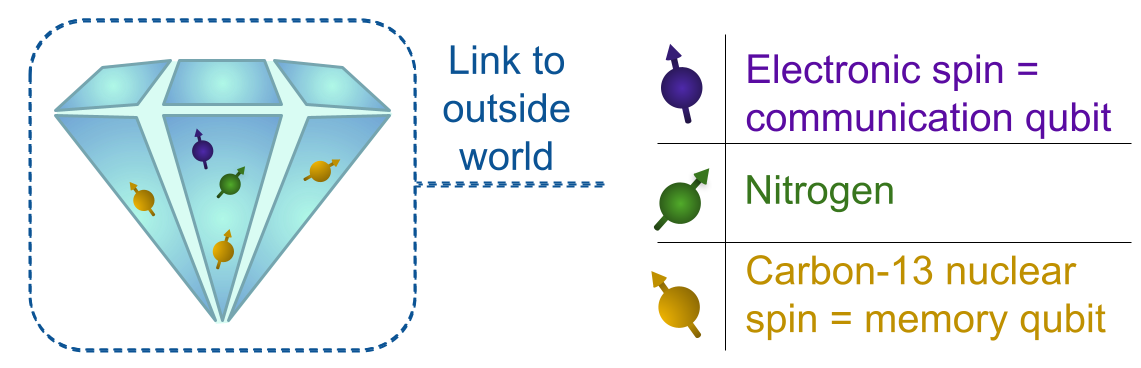}}
\\
\subfloat[the double-NV architecture]{\includegraphics[width=\textwidth]{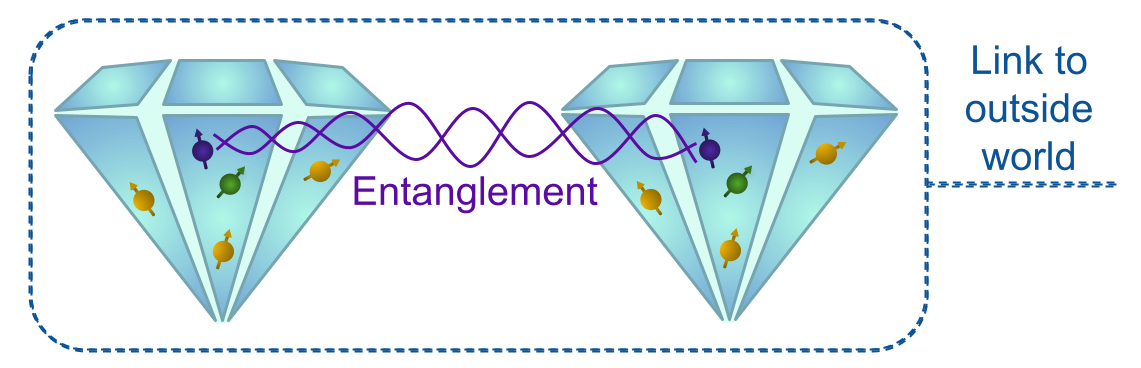}}
\caption{Two NV center in diamond architectures used with distributed quantum applications.}
\label{fig:thetwoarchs}
  \end{minipage}}
  \end{figure}
\fi
\begin{figure}
\centering
\subfloat[single-device architecture]{\includegraphics[width=0.4\textwidth]{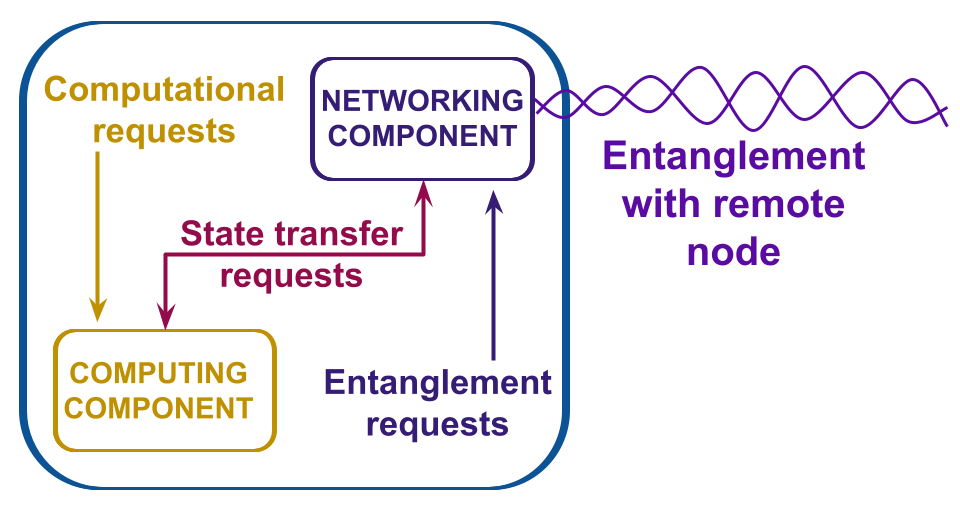}\label{fig:archAbstractSD}}\quad
  \subfloat[double-device architecture]{\includegraphics[width=0.57\textwidth]{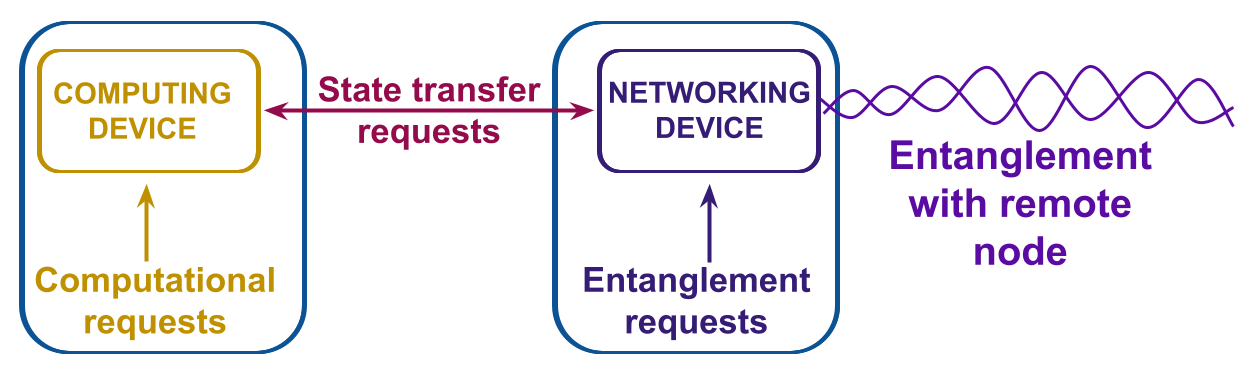}\label{fig:archAbstractDD}}
\caption{Two possible architectures for a quantum processor interfaced to a quantum network: in the first, the processor and the network device are the same device (\emph{single-device (SD) architecture}). This device may have an internal logical or physical division into a computing or networking component. An example of a physical division is the use of a subset of its qubits for networking and others purely for computing. An example of a logical division is a scheduler switching between both functions but networking and computation are performed using the same qubits. 
In the second, two separate devices are used (\emph{double-device (DD) architecture}). 
An application interacts with the system by making three types of requests: local quantum computations (on the computing component/device), network operations (entanglement generation), and movement (state transfer) of generated entanglement into the processor for further processing. The latter requires cooperation from both processing and network devices. 
For SD, a move could be achieved simply by transferring the state to another set of qubits on the same device.
For DD, a move is much more complex, and can be realized \eg, using entanglement generation between the processor and the network devices, followed by teleportation.} 
\label{fig:archAbstract}
\end{figure}
\begin{figure}
\centering
\subfloat[the single-NV architecture]{\includegraphics[width=0.42\textwidth]{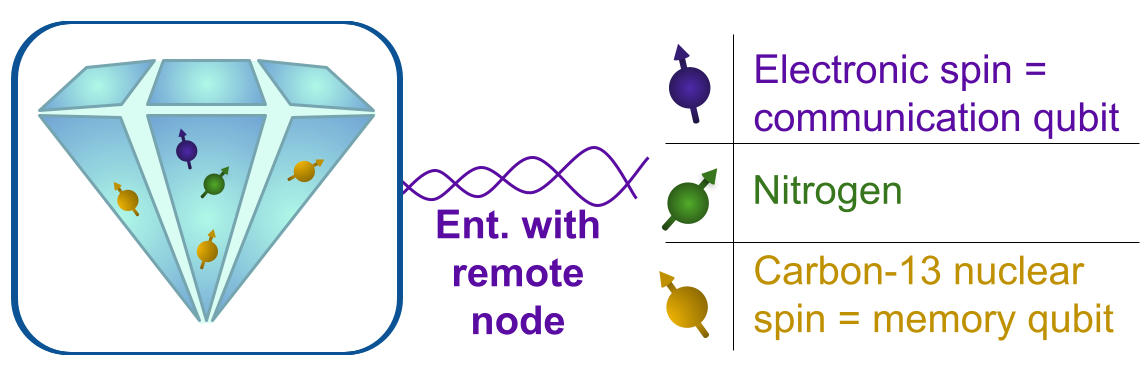}}
\qquad
\subfloat[the double-NV architecture]{\includegraphics[width=0.42\textwidth]{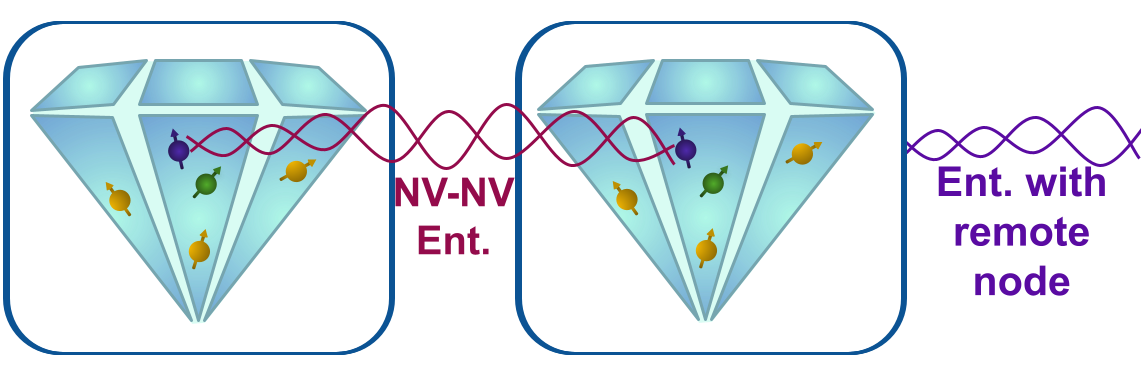}}
\caption{An example of SD and DD architectures based on Nitrogen-Vacancy centers in diamond. Here, the electron spin (purple) acts as an optical interface that can be used as a networking or computing component. Nitrogen (green) and the Carbon-13 spins (yellow) in the surrounding diamond material can be used as computing components. See Appendix \ref{nv_in_diamond_app} for further details about this platform.}
\label{fig:thetwoarchs}
\end{figure}

Quantum communication promises to fundamentally enhance internet technology by enabling application capabilities that are impossible to attain classically. On the one hand, quantum communication could be used to link quantum processors at large distances, enabling quantum internet~\cite{kimble,stagesPaper} applications such as secure communication~\cite{bb84,e91}, improved clock synchronization~\cite{lukin}, or secure delegated quantum computation in the cloud~\cite{blindQC}. On the other hand, quantum communication could connect quantum processors at short distances in order to link several smaller quantum processors together to form one more powerful quantum computing cluster \cite{Jiang_2007}.

To support distributed quantum applications, the architecture of a quantum network node should be capable of two key functions: first, it should enable local quantum computation, \ie, the execution of quantum gates and measurements, at each end node~\cite{stagesPaper} in the network on which applications are run.
Second, it should enable the generation of quantum entanglement between any two nodes in such a network. Entanglement is a special property of the state of two quantum bits (qubits), which cannot be simulated using any form of classical communication between the nodes. A typical quantum network application consists of both local quantum computations and the generation of entanglement, where different applications may have more demand for local quantum processing, or for entanglement generation. An example of an application that is computation-heavy is secure delegated quantum computation~\cite{blindQC}. In contrast, quantum key distribution (QKD)~\cite{bb84,e91,scarani2009security} forms an example of an 
application that is network-heavy, \ie, it is dominated by entanglement generation and the only local operations are measurements. 
Balancing local and networked operations (entanglement generation) can also be important in the efforts to build a quantum repeater~\cite{Briegel_1998,abruzzo2013quantum,bratzik2013quantum,jiang2009quantum}, \ie, a special quantum node that can eventually enable entanglement generation over arbitrarily long distances~\cite{Gisin_2010,Munro_2015}. In this case, proposals for such repeaters employ both local quantum operations (\eg, to perform entanglement purification~\cite{Bennett_1996,Deutsch_1996,dur2007entanglement}), as well as entanglement generation with neighbouring repeater nodes.

When analyzing the performance of quantum networks, one is typically interested in understanding traditional performance metrics, such as the throughput or latency of entanglement generation and local gate execution. Importantly, however, the performance analysis of quantum technologies also demands a characterization of the \emph{quality} of the quantum execution, \ie, how noisy quantum states and operations are. Such a characterization is motivated both by long-term fundamental aspects of quantum applications, as well as the more short-term technological limitations of present-day quantum devices. In the classical world, a system is typically constructed in such a way that all errors are eliminated towards the application~\cite{linkLayerPaper}. That is, an application sees essentially noise-free network transmissions and CPU operations. For many quantum network applications, however, 
noise-free transmission and quantum gate execution are not compulsory.
A good example is QKD~\cite{bb84,e91,scarani2009security}, where noise at the quantum level is dealt with using classical error correction, after measuring the quantum state, in a way that is specific to the application.

In quantum networked systems,  fundamental tradeoffs exist between the quality of the quantum execution, and standard performance metrics such as throughput and latency. A key performance metric in a quantum network is the quality of entanglement (see Section~\ref{sec:fidelity}) being generated between two remote network nodes, where one can choose to trade a higher throughput of entanglement generation, against a lower quality of the resulting entanglement and vice versa~\cite{linkLayerPaper}. On a quantum processor, we furthermore want to understand the quality of a quantum gate's execution, and consequently the quality of the quantum program being executed. If quantum devices were perfect, a quantum gate could be executed with perfect quality, that is, the output is precisely as intended and no noise has occurred. In practice, however, technological limitations mean that gates on all present-day quantum processing platforms are noisy. Such noise can stem from inherent imperfections of the device (constant noise), as well as a time-dependent contribution that depends on the waiting time before the quantum state can undergo further processing.  
The latter form of noise is especially relevant when analyzing networked quantum processors, as is the focus of this work, where we frequently need to wait for a signal from the remote node before processing can continue. However, it also arises when trying to analyze any form of scheduling algorithm on an advanced quantum processor. 
In the quantum literature, the quality of a quantum state is measured by its \emph{fidelity}, and the quality of executing a gate by its \emph{gate fidelity} (see Section~\ref{sec:fidelity}).
Intuitively, the fidelity is a number in the interval $[0,1]$ that measures the closeness of the state (or gate) to a desired target implementation. The larger the fidelity, the closer we are to the target implementation, \ie, the higher the quality of the quantum state (or gate). In this work, we focus on these \emph{quantum performance measures} -- specifically, we study gate fidelity in distributed quantum architectures, as well as the fidelity of entanglement generated by applications that run on them.

Given the need to perform both local quantum operations as well as network operations in order to realize distributed quantum applications, we here consider two different general architectures for interfacing a networked quantum processor to a quantum network (Figure~\ref{fig:archAbstract}). In the first, which we call the single-device (SD) architecture, the same device is used to perform both network operations as well as local quantum computation (Figure~\ref{fig:archAbstractSD}). This is the case in all present-day implementations, such as networked quantum processors based on Nitrogen-Vacancy (NV) centers in diamond~\cite{pompili2021realization}, or Ion Traps~\cite{InnsbruckPhoton50kms}. Abstractly, one can think of these as quantum processors that have two different types of qubits: communication qubits (networking component) with an optical interface for remote entanglement generation, and storage qubits which can only be used for local processing. Limits on experimental control typically prohibit the simultaneous execution of local (two-)qubit gates, and entangling operations. That is, while entanglement generation is in progress, local quantum processing is on hold, and vice versa. The time necessary for local gate execution only depends on the local processing speed. However, the time required for entanglement generation depends on the physical distance to the remote network node. Consequently, in a situation in which the remote node is at a distance, local processing may need to be suspended for a significant amount of time while entanglement generation is in progress.

In the second architecture, we hence consider a scenario in which the system is enhanced by the introduction of a dedicated network device solely used for the purpose of entanglement generation with remote network nodes (Figure~\ref{fig:archAbstractDD}); we refer to this as the double-device (DD) architecture. In this architecture, the network device is linked internally to the processor. While it is not important how this is achieved physically for our general analysis, we provide an example architecture in which both devices are based on NV centers in diamond (Figure \ref{fig:thetwoarchs}) where the internal interface is realized by teleporting the entanglement of the network device into the processor. This requires an additional step of producing entanglement between the network device and the processor to perform the teleportation transfer. Yet, since this entanglement needs to be produced only at very short (on-chip) distances, generation is fast. 
This means that remote entanglement generation via the external networking device and computations on the processor only need to be suspended for a short amount of time when the entanglement is transferred from the former to the latter.
We remark that our analysis is fully general and could also be used to understand physical systems that divide the processor into networking and computing "zones" like segmented ion traps \cite{sangouard2009quantum,pfister2016quantum} or ones that use two different physical systems, such as for example NV centers in diamond for the processor, but a simpler device such as a quantum memory based on atomic ensembles~\cite{RMPensembleRepeaters} as the network device.

When deliberating such architectural choices, several considerations are of concern: first, it is clear that performance may depend on whether we execute a computation-heavy, or a network-heavy application. Indeed, it is clear that in the case of quantum key distribution, where there are no local quantum gates being executed and we simply measure the entanglement right away, the DD architecture may only introduce an unnecessary overhead in implementation. Second, we expect that the performance of both architectures depends on the inherent quality of the quantum devices used to realize them. One key concern is the ability of the quantum device to store quantum states during waiting times: a lower memory lifetime means that waiting times have a much larger impact on the quality of execution. Similarly, the quality of the interface between the processor and the network device is of concern in DD architectures, as it may reduce the quality of the entanglement being transferred.  Finally, while the DD architecture may be of great intuitive appeal, it is much more cumbersome to realize experimentally since one additional device must be constructed. This raises a very practical question as to what achieves more benefit to application performance: implementing the double-device architecture, or investing efforts into improving the quality of the components (\eg, to achieve higher memory lifetimes) in the single-device architecture.

Here, we make the following contributions in analyzing the two architectures:

\renewcommand\labelitemi{$\square$}
\begin{itemize}[leftmargin=*]
\item We provide mathematical formulas for computing the gate and entanglement fidelities for standard quantum noise models. These formulas can be applied to any quantum performance analysis problem, where one would like to understand how a waiting time $t$ affects the quality of quantum gates and entanglement. As such, they allow standard methods from performance analysis that determine the waiting time distribution to be carried over to the quantum domain.
\item For the two architectures introduced above, we determine the most defining characteristics and operational features. We then incorporate these features into a model that is representative of both architectural designs -- specifically, we employ a continuous-time Markov chain (CTMC) to model entanglement generation in a regime where local quantum computation consumes negligible time. This is well motivated in the regime where the distance between processors is large as in a quantum internet, and the time required to produce entanglement dominates with respect to the time to perform local quantum gates. In this case, we obtain analytical expressions for the qubit waiting time distribution, subsequently allowing us to 
apply our fidelity computation method to obtain expressions for the average gate and entanglement fidelities for the two architectures in closed form.
We later relax the assumption that local gates take negligible time, and explore the effects of more time-consuming computation via simulation. The latter is relevant when the processors are physically close.
\item Using the aforementioned analytical results, we determine the strengths and weaknesses of the two architectures. Our analysis can be used to examine general tradeoffs between the quality of the quantum devices, the application behaviour (computation or network-heavy), and the resulting fidelities for quantum gates and entangled states being produced.
In a regime where DD quantum state transfer operations are more noisy than SD ones (\eg, when they rely on imperfect gates), we find that the SD architecture is the more suitable option for applications that are network-heavy, while the DD architecture benefits computation-heavy applications.
While the DD architecture outperforms the SD design in terms of average gate quality, its more complex design makes it harder to implement in practice. We provide sufficient conditions indicating how much the quantum memory lifetime of the components used to implement the SD architecture would need to be improved, in order to achieve the same performance as the DD architecture. 
\item We apply our analytical techniques to evaluate the performance of the two architectures under the assumption that they are realized using the Nitrogen-Vacancy center in diamond platform -- a strong candidate for implementing near-term and future networked quantum nodes \cite{Hensen_2015,Abobeih_2018,Humphreys_2018,Pfaff_2014}.
We explore the effect of state transfer operations on entanglement fidelity and where possible, present the pre-move and post-move entanglement fidelities in closed form. We validate our analytical results with NetSquid \cite{coopmans2020netsquid}, a discrete-event quantum network simulator. 
\end{itemize}

The rest of this paper is organized as follows: in Section \ref{sec:background}, we discuss related work and cover the relevant quantum background. In Section \ref{sec:model}, we introduce the CTMC that is used to model both architectures, and discuss the modeling assumptions. In Section \ref{sec:waitingTimeDistrs} we determine the amount of time that a qubit must spend waiting in storage before it is processed.
This waiting time distribution, along with a noise model, can be used to obtain the average gate and entanglement fidelities -- in Section \ref{sec:fidelityDerivs}, we introduce a method for accomplishing this in a general setting and subsequently apply it to the architectures in Figure \ref{fig:thetwoarchs}. In Section \ref{sec:analyticalEval} we show that when the two architectures have memories of identical quality, the DD architecture always outperforms the SD architecture in terms of average gate fidelity. Interestingly, it is possible that an SD-architecture device with better memories (and thus with a similar performance to a DD device with poorer quality memories) may be the more economical option in terms of manufacturing cost. For this reason, in Section \ref{sec:analyticalEval} we also present sufficient conditions that, if satisfied, ensure the SD outperforms the DD architecture in terms of gate fidelity. In Section \ref{sec:SimNumerObs}, we present analytical and simulation results and make numerical observations in a variety of settings. We make concluding remarks and discuss extensions of the problem in Section \ref{sec:conclusion}.
 
\if{false}
\gv{
Introduce and motivate the problem.
\begin{itemize}
\item[-] Describe the architectures at a high level, talk about use cases. Describe the functionality at a high level: networking vs. computational jobs. Introduce NVs at a high level, motivate why we choose to study this platform.
\item[-] Talk about noise and the effects on quantum states (high-level, intuitive, in a way non-quantum people can understand).
\item[-] Why study these two particular architectures? What are the potential tradeoffs?
\item[-] Why is analysis important (as opposed to just simulation)?
\item[-] Highlight main assumptions: Poisson arrivals, exponentially-distributed service times for all operations except gates, which take negligible time.
\item[-] Can we justify these assumptions? \emph{I.e.}, does our model serve as a good approximation to realistic scenarios (\eg, where gates do not take negligible time)?
\item[-] Contributions: $(i)$ a general framework for computing the gate fidelity; $(ii)$ a CTMC to model both architectures; $(iii)$ validation of analytical results with NetSquid plus simulation of functionality that goes beyond our analytical models.
\item[-] Summary of findings -- when is arch1 better than arch2?
\end{itemize}}
\fi

\section{Background and Related Work}
\label{sec:background}
\subsection{Related Work}
In practice, for any physical platform implementing a networked quantum processor, the quality of quantum gates and states is determined experimentally (see \eg, \cite{Humphreys_2018,Abobeih_2018,Pfaff_2014,Reiserer_2016,Kawakami_2016,Nichol_2017}, among a multitude of others). The objective of these measurements was to characterize one specific setup, but not to explore tradeoffs of potential architectural designs. 
For a string of quantum repeaters with the goal of producing entanglement over long distances, some analytical studies exist that characterize the quality of very specific quantum states, and study their distribution to guarantee a minimum threshold quality see, \eg, \cite{Guha_2015,Shchukin_2019,Brand_2020,Li_2020}. Some analytical studies also exist for the so-called quantum switch, \cite{Vardoyan_2021,Vardoyan_2020,vardoyan2020capacity}, wherein the authors study the maximum possible rate of entanglement switching and the expected number of entangled qubits in storage. These works are very different in spirit since they focus on the creation of quantum entanglement over long distances, and not on tradeoffs between network and computation operations as we do here. We emphasize that in this work we do not assume that the quantum architecture is used for any specific purpose, and make very few assumptions on the physical platforms used to realize potential architectures. Instead, our goal is to abstract these details in the form of configurable modeling parameters, \eg, the demand for entanglement, or the rate at which it is successfully generated. At the time of writing, we are also not aware of a study that considers the interactions between quantum computation and networking within a single system, and how contention for resources and processing time affects fidelity.
\subsection{Quantum background}
\label{sec:quantum_bg}
\subsubsection{Qubits, Quantum States, and Quantum Gates}
Here, we provide the necessary formalism needed in this work, and refer to \eg,~\cite{nielsen&chuang} for a more in-depth introduction.
Quantum information is encoded using quantum bits, or \emph{qubits}, in contrast to the usage of bits in traditional computing.  In addition to holding information in the form of discrete values such as 0 or 1, qubits may hold \emph{quantum states} that are linear combinations of these values.  A (pure) quantum state can be expressed as a vector $|\psi\rangle \in \Complex^d$ of length $1$, where $d$ is often considered to be finite dimensional in quantum technologies. For $n$ qubits the dimension is $d=2^n$, where it is customary to label basis elements of $\Complex^d$ by strings $x=x_1,\ldots,x_n \in \{0,1\}^n$. The Dirac notation $|\cdot\rangle$ is used to represent a vector and is referred to as a \emph{ket} while the conjugate transpose, $\langle \cdot|=|\cdot\rangle^{\dagger}=(|\cdot\rangle^*)^T$, is referred to as a \emph{bra}. 

Quantum information is manipulated through the application of quantum gates. A quantum gate $G$ is represented by a matrix $G \in \Complex^{d\times d}$, where $G$ is unitary, \ie, $G G^\dagger = G^\dagger G = \id$, where $\id$ is an identity matrix of dimension $d$.
Applying a quantum gate $G$ gives us the state $|\psi'\rangle = G|\psi\rangle$.
Common quantum gates include the single-qubit Pauli operators (Figure \ref{fig:thepaulis}).
Given the states of $n$ individual qubits $\ket{\psi_1},\ldots,\ket{\psi_n}$ the joint state of all $n$ qubits is given by the tensor (Kronecker) product of the individual vectors $\ket{\psi} = \ket{\psi_1} \otimes \ldots \otimes \ket{\psi_n}$. Similarly, applying gates $G_1,\ldots,G_n$ to $n$ individual qubits results in the application of the overall operation $G = G_1 \otimes \ldots \otimes G_n$.

Quantum computing and networking applications are realized by applying a series of quantum gates to one or several qubits and then performing a \emph{measurement} of the qubits to read out information in the quantum states. 

\subsubsection{Noisy Quantum States and Operations}\label{sec:noise}
\paragraph{Noisy Quantum States} A convenient way of representing a quantum state $|\psi\rangle$ is as a \emph{density matrix} $\rho=|\psi\rangle\langle\psi|$ which is obtained by taking the outer product of the ket and the bra of the state. 
Importantly, the density matrix formalism allows for the expression of noisy quantum states. For example, a probabilistic process that prepares a desired state $\proj{0}$ with probability $1-p$, but fails and instead prepares $\proj{1}$ with probability $p$, results in a noisy quantum state $\rho = (1-p) \proj{0} + p \proj{1}$. In general, the set of all quantum states on a $d$-dimensional quantum system (including noisy ones) corresponds to the set of matrices
\begin{align}
\mathcal{S} = \left\{\rho \in \Complex^{d\times d}, \rho \succ 0 \mbox{ is positive semi-definite and normalized }\tr(\rho)=1\ \right\}.
\end{align}

\paragraph{Noise Processes}
Using this formalism we can now express the effect of noise on a quantum state. As an example, imagine a noise process that transforms a quantum state $\rho_{\rm initial}$ that is placed into a quantum memory at time $t=0$, into a noisy quantum state $\rho_{\rm noisy}$ after a waiting time $t$. 
Qubits are susceptible to environmental noise that can inadvertently change their quantum state.  Such noise can arise due to imperfect shielding of qubits from external influence as well as imperfect implementations of quantum gates. Mathematically, the set of all possible noise processes corresponds to the set of completely positive trace preserving maps (CPTPM) $\Lambda: \mathcal{S} \rightarrow \mathcal{S}$.
The effect of environmental interaction on quantum states over time is often referred to as \emph{decoherence} and is modeled through the use of \emph{noise models} describing $\Lambda$. Common noise models include $\Lambda = \mathcal{D}$ depolarizing noise,
\begin{align}
\mathcal{D}_t(\rho) = (1-3p)\rho + pX\rho X + pY\rho Y + pZ\rho Z = (1-4p)\rho + 4p \frac{\id}{2}, 
\label{eq:depol}
\end{align}
which drives a quantum state towards the maximally noisy state, also called the maximally mixed state $\frac{\id}{2}$. This state is the quantum equivalent of white noise. Here, time dependence is often expressed by letting $p=\frac{1}{4}(1-e^{-\frac{t}{T}})$, for a fixed $T$ characterizing the quantum memory storing the quantum state. This allows one to express the noise incurred in a quantum memory storing a qubit, after a waiting time of $t$ has elapsed. A model of depolarizing noise is often used as a worst case estimate, when the physical noise process is insufficiently characterized.

In the literature describing implementations of quantum memories, we often have more information about the noise process of the quantum memory device. This noise is generally modeled as dephasing and damping noise (or a combination of both).
Dephasing noise is expressed as
\begin{align}
\mathcal{P}_t(\rho) = (1-p)\rho + pZ\rho Z,
\label{eq:dephase}
\end{align}
where similarly $p=\frac{1}{2}(1-e^{-\frac{t}{T_2}})$ is used to express time-dependence, for a fixed $T_2$ characterizing the memory. This can be understood as an analogue of the classical binary symmetric channel, where a flip operation (here, $Z$) is applied with some probability $p$.

Another common model is the amplitude damping noise channel
\begin{align}
\mathcal{A}_t(\rho) &= M_0 \rho M_0^{\dagger} + M_1 \rho M_1^{\dagger},
\label{eq:damping}
\end{align}
where $M_0,M_1$ have the form
\begin{align}
M_0 = \begin{bmatrix} 1 & 0\\ 0 & \sqrt{1 - \gamma} \end{bmatrix}, \qquad
M_1 = \begin{bmatrix} 0 & \sqrt{\gamma}\\ 0 & 0 \end{bmatrix},
\end{align}
and $\gamma = (1-e^{-\frac{t}{T_1}})$ for a fixed $T_1$ characterizing the effects of the amplitude damping channel. This can be understood as the quantum analogue of a noisy channel of one-sided error which preserves $\proj{0}$, but damps $\proj{1}$ to $\proj{0}$ with an error probability $\gamma$.
 
In most physical implementations of quantum devices, both $\mathcal{P}_t$ and $\mathcal{A}_t$ occur and the noise is described by a composite model
\begin{align}
\mathcal{C}_t(\rho) &= \mathcal{P}_t(\mathcal{A}_t(\rho)) = (1-p)\mathcal{A}_t(\rho) + p Z\mathcal{A}_t(\rho)Z \\
&= (1-p)\left(M_0 \rho M_0^{\dagger} + M_1 \rho M_1^{\dagger}\right) 
+ p Z\left(M_0 \rho M_0^{\dagger} + M_1 \rho M_1^{\dagger}\right)Z \ 
\end{align}
where, in general, $T_2 < T_1$ \cite{nielsen&chuang,tempel2011relaxation,Abobeih_2018,bruzewicz2019trapped,jobez2016towards}. Larger values of $T$, $T_1$, and $T_2$ correspond to a quantum memory with a longer memory lifetime.
 
\paragraph{Noisy Quantum Gates}
The effect of a noisy quantum gate can be described in an entirely analogous fashion. Note that in terms of the density formalism, the effect of applying a gate $G$ on a quantum state can be expressed as 
\begin{align}
G(\rho) = G \rho G^\dagger\ ,
\end{align}
where we follow convention and use $G$ both to denote the unitary matrix, as well as the CPTPM as indicated by context. 
When modeling noise in quantum gates it is customary to model a noisy implementation $\mathcal{E}$ as the ideal gate, followed by possibly time dependent noise $\mathcal{N}_t$. That is, $\mathcal{E} = \mathcal{N}_t \circ G$, where $G$ is the ideal implementation of the gate. We will follow this custom here.
 
As an example, consider a situation in which we perform the gate $G=X$, but then incur a waiting time of $t$ before the next quantum operation is applied. If the total noise process (inherent noise in the gate, plus noise due to waiting) is described by dephasing and damping noise $\mathcal{C}_t$, then the initial state $\rho_{\rm initial}$ is transformed to the noisy state 
\begin{align}
\rho_{\rm noisy} = \mathcal{C}_t\left(X\ \rho_{\rm initial}\ X^\dagger\right)\  
\end{align}
after the waiting time of $t$ has elapsed. 
 
\subsubsection{Entanglement}
Most quantum applications rely on a special property known as \emph{entanglement} that qubits can share. Mathematically, a state $\rho_{AB} \in \Complex^{d_A\times d_A} \otimes \Complex^{d_B \times d_B}$ of a combined quantum system of nodes (or qubits) $A$ and $B$ is called separable if and only if it can be written as a 
classical mixture (\ie, convex combination) of tensor products of single-node states (\ie, $\rho_{AB} = \sum_j p_j \sigma_A^j \otimes \tau_B^j$ for some distribution $\{p_j\}_j$, and states $\{\sigma_A^j\}_j$ on $A$ and $\{\tau_B^j\}_j$ on $B$).
Intuitively, separable states have only classical correlations between $A$ and $B$, since we may toss a coin according to $p_j$ and then prepare individual states on $A$ and $B$ without any form of quantum interaction between them.
 Any state $\rho_{AB}$ that is not separable is called \emph{entangled}.  
In general, $\ket{\Psi} = \frac{1}{\sqrt{d}} \sum_j \ket{j}_A \otimes \ket{j}_B$ is a maximally entangled state between two $d$-dimensional systems $A$ and $B$. 
Such entangled states form the primary building block of most quantum network applications.

\subsubsection{Quantum Quality Measure: Fidelity}\label{sec:fidelity}
 
 \paragraph{Fidelity of a state}
 The fidelity of a quantum state $\rho$ measures how well this state approximates a specific target state $\ket{\psi}$. It is the relevant quantity used to understand how a fixed noise process (\eg, during the preparation of the state) affects its quality, or how the quality of an already prepared state decreases as a function of a waiting time $t$ that this state spends in a quantum memory. 
 Specifically, the fidelity $F$ of a state $\rho$ to a target state $\ket{\psi}$ is defined as~\cite{gilchrist2005distance,van2014quantum}
 \begin{align}
 F = F(\rho,\ket{\psi}) = \langle \psi|\rho|\psi\rangle
 \end{align}
 such that $F=1$ iff $\rho$ is identical to the target state $|\psi\rangle$. 
The fidelity $F$ lies in the interval $[0,1]$ and larger values of $F$ indicate that $\rho$ is closer to the target state $|\psi\rangle$. 

\paragraph{Gate fidelity}
The average gate fidelity measures how well a real-world implementation $\mathcal{E}$ approximates a desired target gate $G$, and is defined as (see \eg,~\cite{nielsen2002simple})
\begin{align}
F_{\rm orig}(\mathcal{E},G) = \int d\Psi \bra{\Psi}G^\dagger \mathcal{E}(\proj{\Psi})G\ket{\Psi}\ ,
\end{align}
where $d\Psi$ is the Haar (uniform) measure on the set of quantum states. That is, the gate fidelity measures how well the implementation approximates the target gate when applied to a specific input state $\ket{\Psi}$, averaged over all possible input states.

When $\mathcal{E} = \mathcal{N}_t \circ G$ (see above) for some time-dependent noise $\mathcal{N}_t$, we will also use the shorthand
\begin{align}
\label{eq:GF}
F(\mathcal{N}_t,G) = F_{\rm orig}(\mathcal{N}_t \circ G,G)
\end{align}
to denote the resulting average fidelity. Eq.~(\ref{eq:GF}) is the relevant quantity when we are interested in the question: provided we had to wait time $t$ after executing the gate $G$ (\eg, due to a scheduling decision), what is the resulting gate fidelity? We remark that since $G$ is unitary, the case where time-dependent noise is applied \emph{before} the execution of the gate $G$ instead reduces to simply studying $\int d \Psi \bra{\Psi}\mathcal{N}_t(\proj{\Psi})\ket{\Psi}$. As we will see, our later formulas apply to both cases.

\paragraph{Entanglement fidelity}
The entanglement fidelity~\cite{originalPaper} measures the quality of an initially maximally entangled state after it was stored in a noisy memory on node $B$ (or $A$) for time $t$. Specifically, for a noise process $\id_A \otimes \mathcal{N}_t$ (no noise on $A$, and time-dependent noise $\mathcal{N}_t$ on system $B$), the entanglement fidelity is defined as 
\begin{align}
F_e(\mathcal{N}_t) = \bra{\Psi}\id_A \otimes \mathcal{N}_t\left(\proj{\Psi}\right) \ket{\Psi}\ , 
\end{align}
where $\ket{\Psi}$ is the maximally entangled state defined above. We remark that the case of noise on $A$ and $B$ can always be dealt with by observing that for any matrix $M$ applied to $A$, this can be translated to applying $M^T$ to $B$: $M_A \otimes \id_B \ket{\Psi} = \id_A \otimes M_B^T \ket{\Psi}$. That is, noise on $A$ and $B$ can be understood by applying both types of noise to the system $B$ in succession.  It turns out that the gate fidelity, and entanglement fidelity are related as~\cite{originalPaper} 
\begin{align}\label{eq:linkFe}
F_{\rm orig}(\mathcal{E},G) = \frac{d F_e(\mathcal{E}) + 1}{d+1}\ ,
\end{align}
where $d$ is the dimension of $A$ and $B$.  For qubits, $d=2$.

\section{Modeling the Architectures}
\label{sec:model}
\if{false}
\begin{table}[]
\centering
\caption{Variables used in the models of the NV architectures and their descriptions.}
\begin{tabular}{ c|l } 
Variable & Description\\
 \hline
$\lambda_e$ & entanglement request arrival rate\\
$\mu_e$ & entanglement generation rate\\
$\lambda_m$ & state transfer request arrival rate\\
$\mu_m$ & state transfer completion rate\\
$\lambda_c$ & computation request arrival rate\\
$\mu_c$ & computation completion rate
\end{tabular}
\label{tab:variables}
\end{table}
\begin{figure}
\centering
\includegraphics[width=0.4\textwidth]{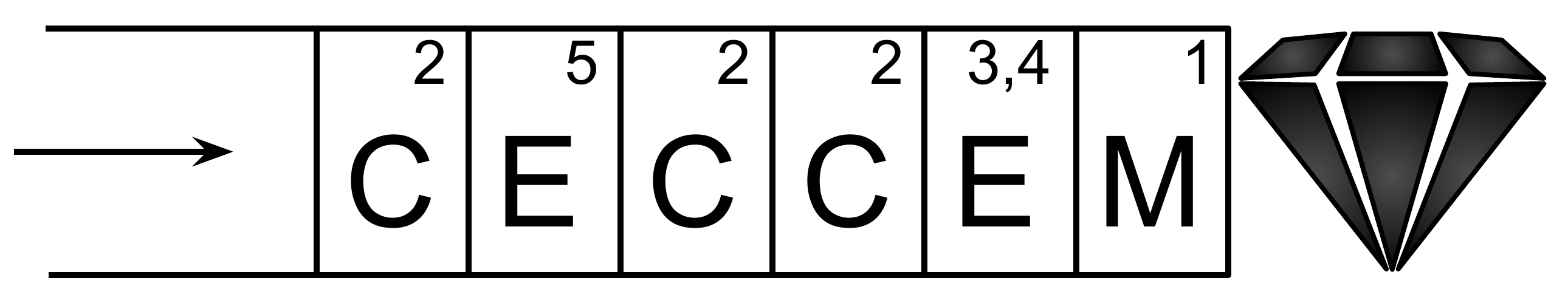}
\caption{Example queue occupancy: $C$, $E$, and $M$ represent computational, entanglement, and moving job, respectively. The numbers in each slot represent the processing order: the moving job is currently being executed, but all computational jobs may be processed before the next entanglement request's processing begins. The second entanglement request is fifth to be processed, since the first entanglement request must be followed by a moving operation (which is fourth to be processed).}
\label{fig:queueEx}
\end{figure}
\fi
\begin{figure}[t]
\centering
\begin{minipage}{0.4\textwidth}
{\begin{align*}
    X=\begin{bmatrix} 0 & 1 \\ 1 & 0\end{bmatrix}&\quad Y=\begin{bmatrix} 0 & -i\\ i & 0\end{bmatrix}\\\\ Z&=\begin{bmatrix} 1 & 0\\ 0 & -1\end{bmatrix}
\end{align*}}
\captionof{figure}{The Pauli operators.}
\label{fig:thepaulis}
\end{minipage}
\begin{minipage}{0.55\textwidth}
\includegraphics[width=\textwidth]{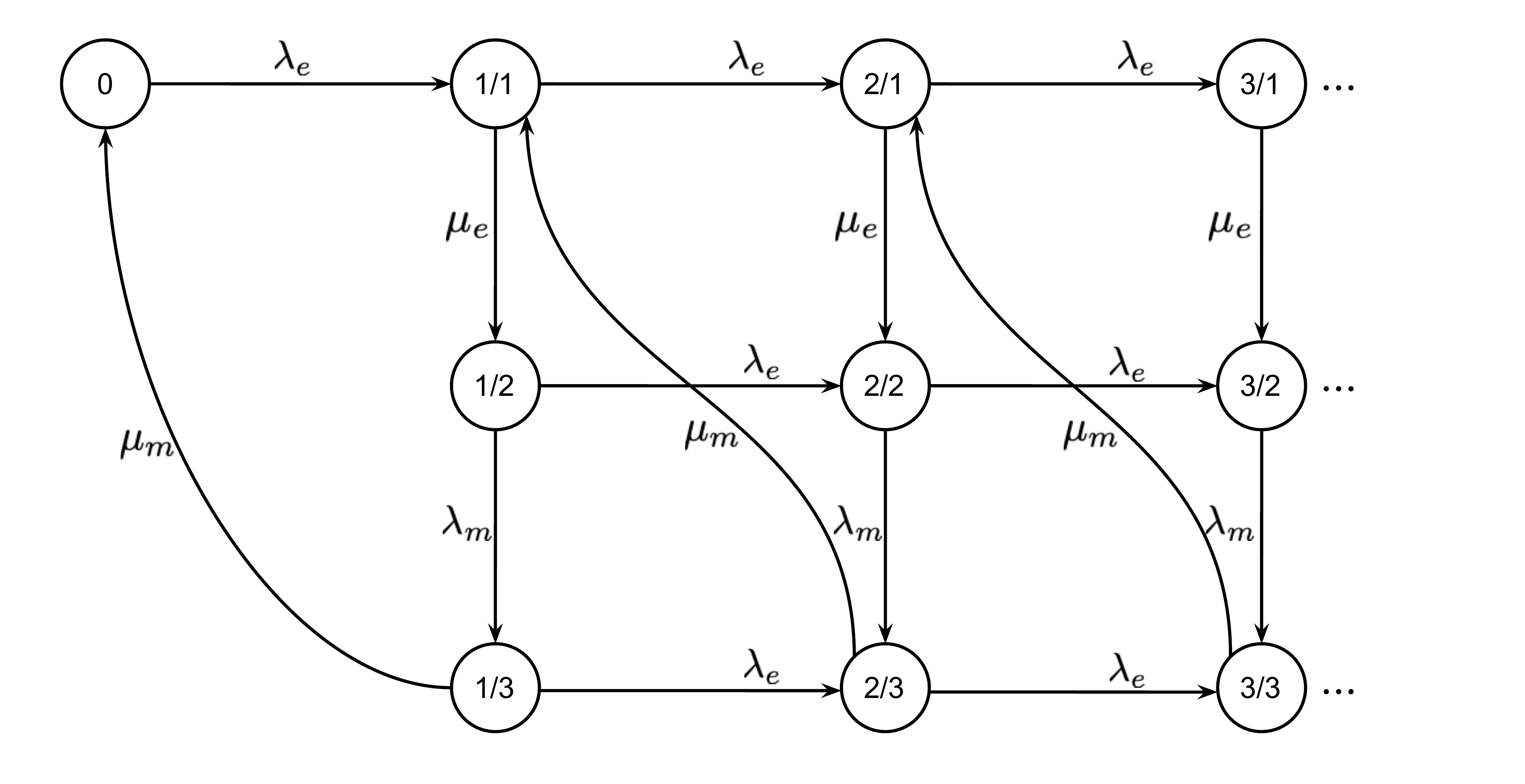}
\caption{A CTMC to model the SD and DD architectures.}
\label{fig:ctmcHypoK1}
\end{minipage}
\end{figure}
We first provide a summary of the architecture attributes to be considered in the modeling and analysis of our problem. First, motivated by the limits of implemented quantum devices, in the SD architecture all operations must be performed sequentially: \eg, computation may not be performed when entanglement generation or a state transfer are in progress. In the DD architecture, entanglement generation and computation (assumed to be independent of each other) may be performed in parallel, as these operations take place in separate devices. When a state transfer operation is in progress, however, both devices in the DD architecture must wait until its completion before servicing another computation or entanglement generation request. This is motivated by the same limit that prevented the simultaneous execution of network and computation operations in the SD architecture: a network operation is needed at the computation device in the DD architecture to transfer the entangled state, but this time only for the time needed to produce entanglement with the very close network device. In both architectures, a state transfer is required before a new entanglement request may be serviced, as this frees up the communication qubit required for further entanglement generation attempts.

To define the state space of our problem, it is helpful to classify the processes that make use of the quantum processors as follows: $(i)$ entanglement generation,
$(ii)$ state transfer operations (we interchangeably refer to these as moving operations), and
$(iii)$ computation. Each class of operations (or jobs) is associated with an arrival rate and a processing rate, as specified in Table~\ref{tab:variables}. We assume that all request arrivals are Poisson and all processing times are exponentially-distributed. $\lambda_e$ represents the demand for entanglement, \ie, the entanglement request rate either from a user or an application. The corresponding parameter $\mu_e$ represents the rate at which (remote) entanglement is generated -- this is a function of the link length. Entanglement generation attempts at the elementary link level are often modeled as Bernoulli trials with some fixed success probability $p_{\rm gen}$ for each attempt -- see, \eg, \cite{Guha_2015,Shchukin_2019,Brand_2020,Vardoyan_2020}. In \cite{Vardoyan_2021}, the authors model the time between successful generation attempts as an exponentially-distributed random variable (r.v.) to accommodate their use of a CTMC when modeling a quantum network node; we adopt their convention here.

When entanglement is successfully generated, we assume that the state of the entangled qubit is eventually moved from the networking component to the computing component for processing, \eg, in the single-NV example, the state is transferred from the electronic spin to the carbon spin. 
Since this moving operation may not be requested immediately, we introduce $\lambda_m$ as the moving request rate. The time to physically perform such moving operations is exponentially-distributed with parameter $\mu_m$; in general $\mu_m$ is lower for DD architectures than for SD ones, as the former requires more complex gate sequences to perform device-device state transfers: \eg, in the double-NV example, NV-NV entanglement generation is required. Finally, computational jobs arrive according to a Poisson process with parameter $\lambda_c$, and their processing times are exponentially-distributed with parameter $\mu_c$. Note that, save for entanglement generation, our use of the exponential distribution in modeling processing times is largely motivated by the resulting simplicity of the model and its analysis.
A more realistic way to model state transfers for an SD design implemented with an NV center in diamond, for instance, would be to assume that their service times are deterministic (albeit, for the double-NV design, the use of the exponential distribution is well-justified due to the need to generate on-chip entanglement when servicing state transfer requests). In practice, the sojourn time distribution in each state is determined not only by the application, but also by the physical platform (\eg, NV center in diamond, ion traps, atomic ensembles). Depending on the latter, it is possible that even the time it takes to perform a local state transfer is a random variable. In the interest of keeping the assumptions as general as possible and the results interpretable, we opt for the exponential distribution. However, if necessary, one may accommodate arbitrary service time distributions by modeling the architecture as a semi-Markov process. Depending on the specific application, it may also be appropriate to include computational jobs as an additional phase of the QBD process. Several other extensions of our model may also be considered depending on the use case -- see Section \ref{sec:conclusion} for further discussion -- but lie outside the scope of this work.
\begin{figure}[t]
{\centering
\begin{minipage}{0.4\textwidth}
  \centering
    \captionof{table}{Variables used in the quantum architecture model and their descriptions.}
\begin{tabular}{ c|l } 
Variable & Description\\
 \hline
$\lambda_e$ & entanglement request arrival rate\\
$\mu_e$ & entanglement generation rate\\
$\lambda_m$ & state transfer request arrival rate\\
$\mu_m$ & state transfer completion rate\\
$\lambda_c$ & computation request arrival rate\\
$\mu_c$ & computation completion rate
\end{tabular}
\label{tab:variables}
\end{minipage}
\qquad\qquad\quad
\begin{minipage}{0.46\textwidth}
  \centering
\includegraphics[width=0.95\textwidth]{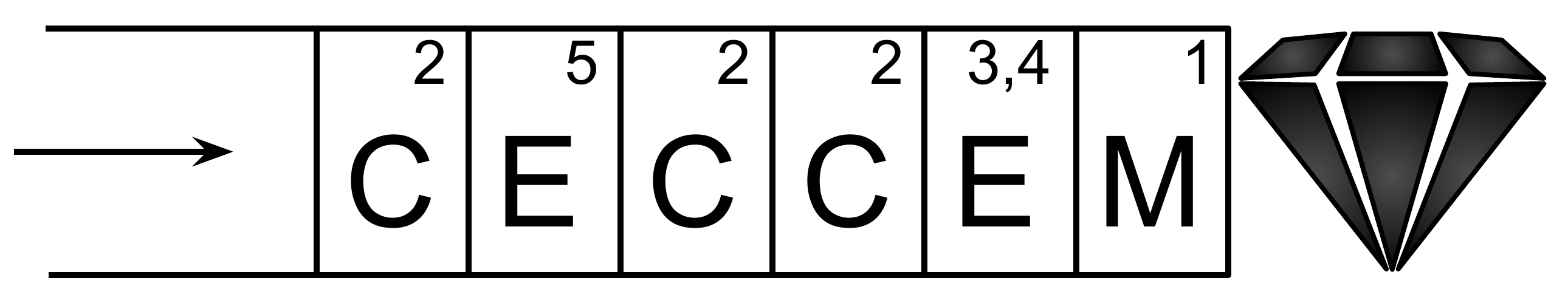}
  \captionof{figure}{Example queue occupancy: $C$, $E$, and $M$ represent computational, entanglement, and moving job, respectively. 
  The numbers in each slot represent the processing order: job $M$ is currently being executed, but all $C$ jobs may be processed before the next $E$ request's processing begins. The second $E$ request is fifth to be processed, since the first $E$ request must be followed by a moving operation.}
\label{fig:queueEx}
\end{minipage}}\end{figure}

In summary, the processing rates $\mu_{e}$, $\mu_m$, and $\mu_c$ depend on the properties of the architecture, while the request arrival rates $\lambda_e$, $\lambda_m$, and $\lambda_c$ depend on application demands. When the $\lambda_e$ value is high or $\mu_e$ is low, the application may be thought of as networking(entanglement)-heavy, while high values of $\lambda_c$ and low values of $\mu_c$ correspond to a computation-heavy application.
In general, $\mu_c$ is much greater than $\mu_e$ and even $\mu_m$ for the DD architecture, as local gates are far less time-consuming than entanglement generation in a network where processors are separated by large distances. For this reason, when constructing the model we make a simplifying assumption that $\mu_c=\infty$, \ie, we assume that computational job processing times are negligible. As a consequence, computational jobs may be processed whenever the processor(s) would otherwise be idle (\ie, waiting for a moving request to arrive), as well as in-between events (\eg, immediately after the completion of a moving request but right before the entanglement generation of the next entanglement job) as their processing does not affect the rest of the system. As a result, a computational job need only to wait for the completion of a single event -- either entanglement generation or a moving request -- and as soon as this event has been completed, all computational jobs that are in the queue may be processed instantaneously\footnote{For some physical platforms, additional simplifying assumptions on computational jobs may be necessary -- see Appendix \ref{nv_in_diamond_app} for a discussion.}. An example is shown in Figure \ref{fig:queueEx}. These modeling assumptions ultimately allow us to obtain all necessary performance measures in closed form; however, in Section \ref{sec:SimNumerObs} we remove the assumption on negligible processing times for computational jobs and observe the effects on the average fidelity numerically. Lower values of $\mu_c$ may be used to model ``atomic'' gate sequences which must not be interrupted by any other task or operation.

With the aforementioned assumptions, we may model both the SD and DD architectures as $M/HYPO_3/1$ queueing systems, where the arrivals correspond to entanglement requests and are a Poisson process with rate $\lambda_e$, and the service times are hypo-exponentially distributed with three service stages each of which are exponentially-distributed with parameters $\mu_e$, $\lambda_m$, and $\mu_m$. Figure~\ref{fig:ctmcHypoK1} depicts the CTMC representing this queueing system: each state of the form $N/k$ corresponds to $N$ outstanding entanglement requests in the system, with the first job in the $k$th stage of its processing. Entanglement requests are processed according to a first-in, first-out (FIFO) policy: the first stage ($k=1$) is entanglement generation, the second ($k=2$) is awaiting the arrival of a moving request, and the third ($k=3$) is the execution of a moving request. Note that the next entanglement request cannot begin processing until all three stages of the previous request have been completed, since the communication qubit must be freed before entanglement generation may be attempted again. State $0$ corresponds to the case with no outstanding entanglement requests.

Note from the CTMC that a computational request only waits whenever its arrival coincides with the system being in states of the form $N/1$ and $N/3$ in the SD architecture, or if the system is in a state $N/3$ in the DD architecture. Thus, by the memoryless property of the exponential distribution, a computational job's waiting time distribution is given either by $f_e (t) = \mu_e e^{-\mu_e t}$ or by $f_m(t) = \mu_m e^{-\mu_m t}$, depending on the state of the system upon arrival. 
We may obtain the probabilities of arrival into states $N/1$ and $N/3$ from the stationary distribution of the CTMC. We remark that when computing the entanglement fidelity, we are interested in the amount of time a newly-entangled qubit must wait for a state transfer request to arrive, \ie, we require the \emph{waiting} time distribution, and not the sojourn time distribution of that qubit (which also includes the amount of time it takes to perform the moving operation). The reason is that the specific gate sequence used to perform the state transfer in a given architecture already implicitly accounts for the the time it takes to execute the gates.

A note on mathematical notation: in the remainder of the paper, we use superscripts $^{(1)}$ and $^{(2)}$ to denote parameters corresponding to the SD and DD architectures, respectively. \emph{E.g.}, $\mu_m^{(1)}$ corresponds to the moving rate in the SD architecture, while $T_1^{(2)}$ and $T_2^{(2)}$ refer to the memory lifetimes of the DD architecture.

\section{Waiting Time Distributions}
\label{sec:waitingTimeDistrs}
The CTMC shown in Figure \ref{fig:ctmcHypoK1} has been studied in literature: specifically, it is a quasi birth-death (QBD) process with the special property that one of the blocks in its generator matrix is a rank 1 matrix, allowing us to compute the rate matrix explicitly using the results in \cite{latouche1999introduction}. For completeness, we include the rate matrix derivation and ergodicity condition for this Markov chain in Appendix \ref{ctmc_analysis_app}. For the following computations, assume that the mean drift condition
\begin{align}
\frac{1}{\lambda_e} > \frac{1}{\mu_e}+\frac{1}{\lambda_m}+\frac{1}{\mu_m}
\label{eq:meandriftcondHypo}
\end{align}
is satisfied so that a stationary distribution exists. Recall from the discussion in Section \ref{sec:model} that to determine the waiting time distribution of a computational job, it suffices to compute the stationary probabilities of states $N/1$ and $N/3$, which we label $\pi_{N/1}$ and $\pi_{N/3}$, respectively. Specifically, rather than having to derive the individual stationary probability of each state in phase 1 or 3, we need only to compute the aggregate probabilities $\sum\limits_{N=1}^{\infty}\pi_{N/1}$ and $\sum\limits_{N=1}^{\infty}\pi_{N/3}$,
since a computational request's waiting time has no dependence on $N$, but only on the phase (stage) in which the QBD process has been found upon arrival. In Appendix \ref{app:stationary}, we show that $\sum\limits_{N=1}^{\infty}\pi_{N/1} = \lambda_e/\mu_e$ and $\sum\limits_{N=1}^{\infty}\pi_{N/3} = \lambda_e/\mu_m$.

We are now ready to compute the waiting time distributions of computational jobs in both types of architectures. Recall that in the SD architecture, a computational job may be processed immediately, as long as there is no ongoing entanglement generation or moving job in progress. Specifically, if a computational request arrives while the QBD process is in a state $N/2$, then it is processed immediately (the waiting time is zero), and otherwise, the request is queued behind the ongoing entanglement or moving request. Thus, the waiting time (conditioned on the stage $k$) for a computational job in the SD architecture is given by $W_1 = S_e\mathds{1}_{\{k=1\}} + S_{m_1}\mathds{1}_{\{k=3\}}$, where $\mathds{1}$ is the indicator function,
\if{false}
\begin{align}
W_1 &= \begin{cases}
S_e,& \text{if } k=1,\\
0, &\text{if } k=2,\\
S_{m_1}, &\text{if } k=3,
\end{cases}
\end{align}
where 
\fi
$S_e$ is an exponentially-distributed r.v. with mean $1/\mu_e$ and probability density function (p.d.f) $f_e(t)$, and $S_{m_1}$ is an exponentially-distributed r.v. with mean $1/\mu_{m}^{(1)}$ and p.d.f. $f_{m_1}(t)$, with $\mu_m^{(1)}$ the rate of completing a moving request (after its arrival, meaning that this is the rate of solely executing the gates to fulfill a moving request) in the SD architecture. 
Then, by the law of total probability (\emph{cf.} Eq. (2.26) in \cite{pinsky2010introduction}), the marginal p.d.f. for the waiting time of computational jobs in the SD architecture is given by
\begin{align}
f_{W_1}(t) &= \sum\limits_{N=1}^{\infty}\pi_{N/1}f_e(t) + \sum\limits_{N=1}^{\infty}\pi_{N/3}f_{m_1}(t) + \left(1-\sum\limits_{N=1}^{\infty}\pi_{N/1} - \sum\limits_{N=1}^{\infty}\pi_{N/3}\right)\delta(t)\\
&=\lambda_e e^{-\mu_e t} + \lambda_e e^{-\mu_m^{(1)} t}+ \left(1-\frac{\lambda_e}{\mu_e} - \frac{\lambda_e}{\mu_m^{(1)}}\right)\delta(t),
\label{eq:wdistr_arch1}
\end{align}
where $\delta(t)$ is the Dirac delta function, defined as
\begin{align*}
\delta(t) = \begin{cases}
+\infty, & t = 0,\\
0, & t\neq 0,
\end{cases}
\quad\text{and satisfies}\quad
\int\limits_{-\infty}^{\infty}\delta(t) dt = 1.
\end{align*}
Next, we consider the waiting time distribution of computational jobs in the DD architecture, wherein a computational request need only wait if a moving request is actively being executed. Letting $f_{m_2}$ be the p.d.f. of an $\sim Exp(\mu_m^{(2)})$ r.v., this distribution is given by
\begin{align}
f_{W_2}(t) &= \sum\limits_{N=1}^{\infty}\pi_{N/3}f_{m_2}(t) +\left(1-\sum\limits_{N=1}^{\infty}\pi_{N/3}\right)\delta(t) =  \lambda_e e^{-\mu_m^{(2)} t}
+\left(1-\frac{\lambda_e}{\mu_m^{(2)}}\right)\delta(t).
\label{eq:wdistr_arch2}
\end{align}
When using $f_{W_1}$ and $f_{W_2}$ to compute average gate fidelity in Section \ref{sec:fidelityDerivsAppToOurProb}, we will use the definition
\begin{align}
\int\limits_{a}^{b} f(x)\delta(x-c)dx = f(c),
\label{eq:diracDeltaIntegral}
\end{align}
where $a\leq c\leq b$ and $f(x)$ is a function continuous on the interval $[a,b]$.

Another quantity of interest is the waiting time distribution of a newly-entangled qubit while it awaits a state transfer request. For both architectures, this is given by $g(t) = \lambda_m e^{-\lambda_mt}$.
\if{false}
\begin{align}
    g(t) = \lambda_m e^{-\lambda_mt}.
    \label{eq:wdistr_move}
\end{align}
\fi

\section{Formulas for computing quantum fidelities}
\label{sec:fidelityDerivs}
We now provide several general formulas that can be used to link an understanding of waiting times $t$ to the gate and entanglement fidelities for the standard noise models used to describe quantum devices. Our formulas for the gate fidelities can be applied to the situations described in Section~\ref{sec:fidelity}, where we need to wait for a time $t$ before or after executing a quantum gate. In this work, this waiting time occurs since we need to suspend quantum processing when performing network operations (see Section~\ref{sec:model}). However, we remark that our formulas are applicable to any other situations where such waiting times arise, such as for example in the analysis of algorithms for scheduling gates on a quantum processor. 
We emphasize that our formulas also apply to a situation where one would simply want to understand the average reduction in quality when storing a qubit in memory, which corresponds to applying the trivial gate $G=\id$.

Using the known link between gate and entanglement fidelities~\cite{originalPaper} in Eq.~\eqref{eq:linkFe}, these formulas can also be directly applied to understand how the quality of an entangled link decays as a function of a waiting time $t$. The use case in this work is to understand the decay of entanglement due to waiting for gate or move operations (see Section~\ref{sec:model}). However, it can also be applied to any other situation where a waiting time $t$ is incurred before the entanglement can be processed.
\begin{lemma}
Let $\mathcal{D}_t$, $\mathcal{P}_t$, $\mathcal{A}_t$ and $\mathcal{C}_t$ denote the depolarizing channel, dephasing channel, amplitude damping and composite channel for a single qubit respectively, as defined in Section~\ref{sec:noise}. Let $t > 0$\, and $G$ be any quantum gate (unitary).
\begin{itemize}
\item[-] For depolarizing noise $\mathcal{D}_t$ with $p = \frac{1}{4}\left(1 - e^{-t/T}\right)$\ ,
\begin{align}
F(\mathcal{D}_t, G) = \frac{1}{2}\left(1 + e^{-\frac{t}{T}}\right)\ .
\end{align}
\item[-] For dephasing noise $\mathcal{P}_t$ with $p=\frac{1}{2}\left(1-e^{-t/T_2}\right)$\ , 
\begin{align}
F(\mathcal{P}_t, G) = \frac{1}{3}\left(2 + e^{-\frac{t}{T_2}}\right)\ .
\label{eq:dephasing_fidelity}
\end{align}
\item[-] For amplitude damping noise $\mathcal{A}_t$ with $\gamma = (1-e^{-\frac{t}{T_1}})$,
\begin{align}
F(\mathcal{A}_t, G) = \frac{1}{6}\left(3 + e^{-\frac{t}{T_1}} + 2e^{-\frac{t}{2T_1}}\right)\ .
\label{eq:amp_damp_fidelity}
\end{align}
\item[-] For the composite noise model $\mathcal{C}_t$ using $\gamma = (1-e^{-\frac{t}{T_1}})$ for $\mathcal{A}_t$ and $p=\frac{1}{2}\left(1-e^{-t\left(\frac{1}{T_2} - \frac{1}{2T_1}\right)}\right)$ for $\mathcal{P}_t$ (to account for dephasing effects of $\mathcal{A}_t$ \cite{nielsen&chuang,coopmans2020netsquid})\ ,
\begin{align}
F(\mathcal{C}_t, G) = \frac{1}{6}\left(3 + e^{-\frac{t}{T_1}} + 2e^{-\frac{t}{T_2}}\right),
\end{align}
where we assume $0 < T_2 \leq 2T_1$ so that Eqs.~\eqref{eq:dephasing_fidelity} and \eqref{eq:amp_damp_fidelity} are recovered as $T_1 \to \infty$ and $T_2 \to 2T_1$.
\end{itemize}
\end{lemma}
To the non-quantum expert, it may come as a surprise that the formulas above only depend on the noise, but not on the specific gate $G$. As we will see, this is simply a consequence of $G$ being unitary, combined with the fact that we uniformly average over the set of all quantum states and this average does not change when a unitary is applied.

Given an understanding of the waiting time distribution, one may then readily compute the average gate fidelity due to waiting as
\begin{align}
F_{avg}(\mathcal{N}_t,G) = \int dw(t) F(\mathcal{N}_t,G)\ ,
\label{eq:Favg}
\end{align}
where $dw(t)$ is the measure over waiting times $t$ resulting from a specific model, and $\mathcal{N}_t$ is the quantum channel.
\subsection{Derivation}
We remark that there are several methods to obtain the same result, and for completeness we present a self-contained derivation using only elementary facts from quantum information theory in Appendix ~\ref{app:avg_fidelity_derivation}. Here, we make use of a result in ~\cite{bowdrey2002fidelity} for the case where the quantum channels are applied to qubits, as relevant for the standard noise models above.
From~\cite{bowdrey2002fidelity} we have that for any one qubit quantum channel $\mathcal{E}$ used to approximate a gate $G$ that
\begin{align}
F_{\rm orig}(\mathcal{E}, G) = \frac{1}{2} + \frac{1}{12}\left(\Tr\left[GXG^{\dagger}\mathcal{E}\left[X\right] \right] + \Tr\left[GYG^{\dagger}\mathcal{E}\left[Y\right] \right] + \Tr\left[GZG^{\dagger}\mathcal{E}\left[Z\right] \right]\right)\label{eq:bowen}
\end{align}
where $X$, $Y$ and $Z$ are the Pauli matrices defined in Section~\ref{sec:background}. When $\mathcal{E} = G \circ \mathcal{N}_t$ (\ie, noise in the gate is modeled by applying first a noisy channel $\mathcal{N}_t$ followed by the ideal implementation of $G$) we have that since $G$ is unitary ($G^\dagger G = G G^\dagger = \id$),
\begin{align*}
    F_{\rm orig}(G \circ \mathcal{N}_t,G) = \int d\psi \bra{\psi}G^\dagger G \mathcal{N}_t(\proj{\psi}) G^\dagger G \ket{\psi} = \int d\psi \bra{\psi}\mathcal{N}_t(\proj{\psi})\ket{\psi} = F_{\rm orig}(\mathcal{N}_t,\id)\ .
\end{align*}
Using~\eqref{eq:bowen}, the Pauli matrices, and the definitions of the quantum noise channels (see Section~\ref{sec:background}), matrix algebra then yields the claimed formulas above.
For the case where $\mathcal{E} = \mathcal{N}_t \circ G$ (\ie, the noisy gate is modeled by first applying the ideal gate $G$ and then applying a noise process $\mathcal{N}_t$, a common convention in quantum technologies), we also obtain
\begin{align*}
    F_{\rm orig}(\mathcal{N}_t \circ G,G) = \int d\psi \bra{\psi}G^\dagger \mathcal{N}_t(G \proj{\psi} G^\dagger) G \ket{\psi} = \int d\psi \bra{\psi}\mathcal{N}_t(\proj{\psi})\ket{\psi} = F_{\rm orig}(\mathcal{N}_t,\id)\ ,
\end{align*}
where this time we have made use of the fact that the Haar (uniform) measure $d\psi$ on the set of quantum states is invariant under the application of a unitary $G$, since a unitary simply permutes the set of quantum states $G\ket{\psi} = \ket{\psi'}$. We remark that in quantum information operations do not generally commute and $\mathcal{N}_t \circ G \neq G \circ \mathcal{N}_t$ for most choices of $G$ and $\mathcal{N}_t$.



\subsection{Application to our problem}
\label{sec:fidelityDerivsAppToOurProb}
As discussed in Section \ref{sec:quantum_bg}, fidelity acts as a measure by which we may evaluate the performance of single- and double-device architectures.  For computation requests, we determine an associated gate fidelity that reflects the quality of the quantum gate(s) that are applied in the computation.  For moving requests, we consider the fidelity of the entanglement that is delivered to applications. 

We evaluate the average gate fidelity (Eq. (\ref{eq:Favg})) for computation requests on the SD and DD architectures using the composite noise model $\mathcal{C}$ and the waiting time distributions in (\ref{eq:wdistr_arch1}), (\ref{eq:wdistr_arch2}):
\begin{align}
F^{(1)}_{avg}(\mathcal{C}, G) &= \int_0^{\infty} f_{W_1}(t)F(\mathcal{C}, G)dt\nonumber \\
&= 1 - \frac{\lambda_e}{2}\left(\frac{1}{\mu_e} + \frac{1}{\mu_m^{(1)}}\right) + \frac{\lambda_e}{6}\left(\frac{2T_2}{\mu_e T_2 + 1} + \frac{T_1}{\mu_e T_1 + 1} + \frac{2T_2}{\mu^{(1)}_m T_2 + 1} + \frac{T_1}{\mu^{(1)}_m T_1 + 1}\right),
\label{eq:Favg1_composite}\\
F^{(2)}_{avg}(\mathcal{C}, G) &= \int_0^{\infty} f_{W_2}(t)F(\mathcal{C}, G)dt
= 1 - \frac{\lambda_e}{2\mu_m^{(2)}} + \frac{\lambda_e}{6}\left(\frac{2T_2}{\mu^{(2)}_m T_2 + 1} + \frac{T_1}{\mu^{(2)}_m T_1 + 1}\right),
\label{eq:Favg2_composite}
\end{align}
where we use (\ref{eq:diracDeltaIntegral}) to evaluate the integrals above.
We also evaluate the average fidelity of the entanglement that is moved into memory. The waiting time distribution for moving requests for both architectures is given by $g(t) = \lambda_m e^{-\lambda_m t}$, so that
\begin{align*}
F_e(\mathcal{C}) =
\frac{3F_{avg}(\mathcal{C}, I)-1}{2} &=   \frac{1}{4}\int_0^{\infty} \lambda_m e^{-\lambda_m t}\left(3 + e^{-\frac{t}{T_1}} + 2e^{-\frac{t}{T_2}}\right)dt -\frac{1}{2}
= \frac{1}{4} + \frac{\lambda_m}{4}\left(\frac{T_1}{\lambda_m T_1 + 1} + \frac{2 T_2}{\lambda_m T_2 + 1}\right).
\end{align*}

\section{Analytical Evaluation}
\label{sec:analyticalEval}
In Section \ref{sec:waitingTimeDistrs}, we derived the waiting time distributions of computational jobs in each of the architectures, and in Section \ref{sec:fidelityDerivs} we determined how these distributions translate into the average gate fidelity. Using these results, it is possible to compare the average gate fidelity of the SD to that of the DD architecture. Specifically, we will show that when computational job processing times are negligible (\ie, when the queueing systems are both represented by the CTMC in Figure \ref{fig:ctmcHypoK1}) and the memories in both architectures are of identical manufacturing quality, then the DD architecture always outperforms the SD architecture in terms of the average gate fidelity. Our standard assumption that $\mu_e\leq\mu_m^{(2)}$ holds for the following discussion.

\begin{proposition}
\label{prop:avgGFcomp}
When $\mu_c=\infty$, the mean drift conditions (\ref{eq:meandriftcondHypo}) for the SD and DD architectures are satisfied, and $F(\mathcal{N}_t,G)$ is identical for both architectures, the DD architecture yields a higher average gate fidelity than the SD architecture.
\end{proposition}
See Appendix \ref{app:prop2Proof} for a proof of Proposition \ref{prop:avgGFcomp}.

It is worth emphasizing that the result in Prop. \ref{prop:avgGFcomp} holds when the two Markov chains modeling the architectures are stable; \ie, while the DD architecture yields better performance for average gate fidelity, \cf (\ref{prop:avgGFcomp}) it is also more difficult to ensure its stability, since in general, $\mu_m^{(1)} > \mu_m^{(2)}$.

For the remainder of this section, we focus on the composite noise model for storage introduced in Section \ref{sec:fidelityDerivsAppToOurProb}.
Now, suppose that the characteristic memory times of the two architectures are not identical. As discussed previously, in such cases it is possible that an SD architecture with higher-quality memories is the more cost-effective option while also yielding a higher average gate fidelity than a DD architecture with memories of poorer quality. This brings up a natural question: given a DD architecture with fixed characteristic memory times $T_1^{(2)}$ and $T_2^{(2)}$, what conditions must the memory lifetimes of an SD architecture satisfy in order to outperform the former in terms of average gate fidelity? From Eqs. (\ref{eq:Favg1_composite}) and (\ref{eq:Favg2_composite}), we see that $F_{avg}^{(1)} > F_{avg}^{(2)}$ when
{\footnotesize
\begin{align}
\frac{1}{\mu_e} + \frac{1}{\mu_m^{(1)}}
- \frac{1}{3}\left(\frac{2T_2^{(1)}}{\mu_e T_2^{(1)} + 1} + \frac{T_1^{(1)}}{\mu_e T_1^{(1)} + 1} + \frac{2T_2^{(1)}}{\mu^{(1)}_m T_2^{(1)} + 1} + \frac{T_1^{(1)}}{\mu^{(1)}_m T_1^{(1)} + 1}\right) 
<
\frac{1}{\mu_m^{(2)}}- \frac{1}{3}\left(\frac{2T_2^{(2)}}{\mu^{(2)}_m T_2^{(2)} + 1} + \frac{T_1^{(2)}}{\mu^{(2)}_m T_1^{(2)} + 1}\right).
\label{eq:compositeCond}
\end{align}}
For the following discussion, it is useful to keep in mind that for a constant $c>0$, $\lim\limits_{x\to\infty} x/(cx+1) = 1/c$. This implies that both sides of the inequality above are positive.
Recall from previous discussion that the remote entanglement generation rate $\mu_e$ is a smaller value than the moving rates $\mu_m^{(1)}$ and $\mu_m^{(2)}$. From (\ref{eq:compositeCond}), it is easy to see that the value of $\mu_e$ plays a significant role in determining how much the memory lifetimes must compensate to improve the performance of the SD architecture. 

To obtain a more interpretable and intuitive understanding of how high the SD architecture memory times must be, we derive a sufficient condition in terms of solely $T_2^{(1)}$, $T_1^{(2)}$, $\mu_e$, $\mu_m^{(1)}$, and $\mu_m^{(2)}$. This condition serves as a good bound to (\ref{eq:compositeCond}) when the memory lifetimes $T_1^{(2)}$ and $T_2^{(2)}$ are not too far apart from each other, and becomes tighter as $\mu_e$ increases.
\begin{proposition}
\label{prop:compositeBound}
Assume $\mu_e$ is greater than one and $\mu_e < \mu_m^{(2)}<\mu_m^{(1)}$.
Then the SD architecture on average achieves a higher gate fidelity than the DD architecture when
\begin{align}
T_2^{(1)} >\frac{\mu_m^{(2)}(\mu_m^{(2)}T_1^{(2)}+1)(\mu_e+\mu_m^{(1)})}{\left(\mu_e\mu_m^{(1)}\right)^2} - \frac{\mu_e+\mu_m^{(1)}}{\sqrt{2}\mu_e\mu_m^{(1)}}. 
\label{eq:compositeBound}
\end{align}
\end{proposition}
See Appendix \ref{app:boundProof} for a proof of this proposition.
\if{false}
\begin{lemma}
\label{lemma:SecondLemma}
For any positive reals $a$, $b$, $c$, and $x$, if $a>b>1$, then
\begin{align}
\frac{1}{a(ax+1)}+\frac{1}{b(bx+1)} < \frac{a+b}{\sqrt{ab}(\sqrt{ab}x+1)}.
\label{eq:2ndlemmaStatement}
\end{align}
\end{lemma}
\begin{proof}{(Lemma \ref{lemma:SecondLemma})}
Assume for a contradiction that (\ref{eq:2ndlemmaStatement}) is not true. Then it must be that
\begin{align}
\frac{1}{a(ax+1)}+\frac{1}{b(bx+1)} &\geq \frac{a+b}{\sqrt{ab}(\sqrt{ab}x+1)},\\
\frac{(a^2+b^2)x+a+b}{ab(ax+1)(bx+1)} &\geq  \frac{a+b}{abx+\sqrt{ab}},\\
((a^2+b^2)x+a+b)(abx+\sqrt{ab}) &\geq ab(a+b)(ax+1)(bx+1),\\
ab(a^2+b^2)x^2 +(a^2b+ab^2+\sqrt{ab}(a^2+b^2))x + (a+b)\sqrt{ab} &\geq ab(a+b)(abx^2 + (a+b)x+1).
\label{eq:lemmaContr}
\end{align}
Next, consider the coefficients of the $x^2$ terms on each side of the inequality above.
On the left-hand side, we have $ab(a^2+b^2) = a^3b+ab^3$, which is less than $a^3b^2+a^2b^3$ on the right-hand side (recall that $a>b>1$). Next, note that the term $\sqrt{ab}(a+b)$ on the left-hand side is less than the term $ab(a+b)$ on the right side. Finally, consider the coefficients of the $x$ terms on each side of the inequality. On the left, we have
\begin{align*}
L = a^2b+ab^2+a^2\sqrt{ab}+b^2\sqrt{ab},
\end{align*}
while on the right side we have
\begin{align*}
R = ab(a+b)^2 = ab(a^2+2ab+b^2) = a^3b + 2a^2b^2 + ab^3.
\end{align*}
Note that $a^3 b > a^2\sqrt{ab}$ and $ab^3 > b^2\sqrt{ab}$. Finally, $2a^2b^2 > a^2b +a b^2$ since $a^2b^2 > a^2b$ and $a^2b^2 > ab^2$.
Together, these facts contradict (\ref{eq:lemmaContr}), proving the lemma.
\end{proof}
\fi

Condition (\ref{eq:compositeBound}) tells us that the faster the entanglement generation rate, the smaller the memory requirements are on the SD architecture to ensure that it outperforms the DD architecture. This is intuitive since in the SD architecture, entanglement generation is the most time-consuming and therefore the most detrimental operation to the gate fidelity. Condition (\ref{eq:compositeBound}) also tells us that the faster the state transfer requests are executed in the DD architecture, the better memories are required for the SD architecture -- an intuitive consequence of the fact that the most detrimental operations to gate fidelity in the DD architecture are moving requests.

\section{Simulation and Numerical Observations}
\label{sec:SimNumerObs}
\begin{figure}[t]
\centering
\subfloat[$\mu_e=10$, $\lambda_e=1$]{\includegraphics[width=0.4\textwidth]{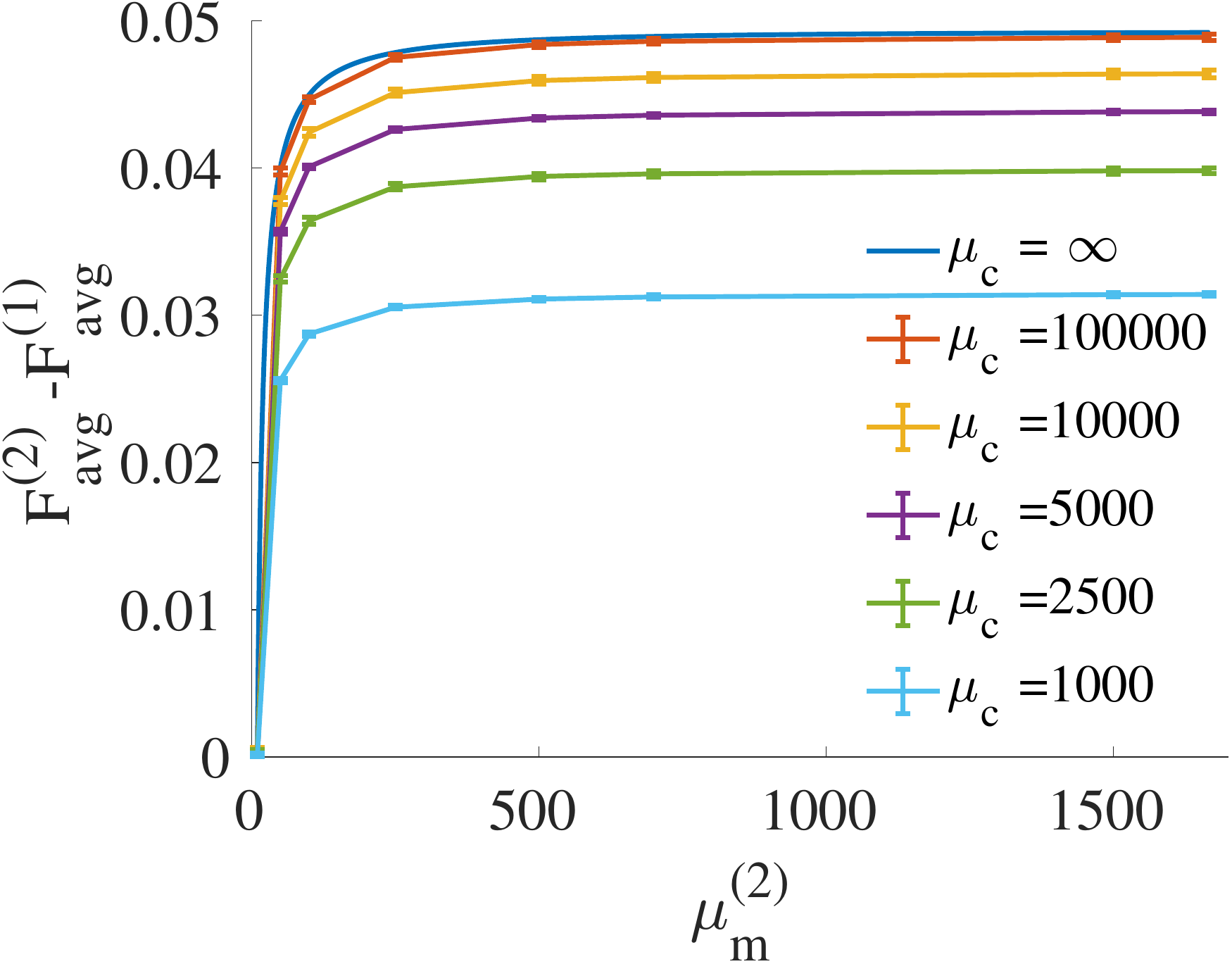}}\qquad\qquad
\subfloat[$\mu_e=500$, $\lambda_e=50$]{\includegraphics[width=0.4\textwidth]{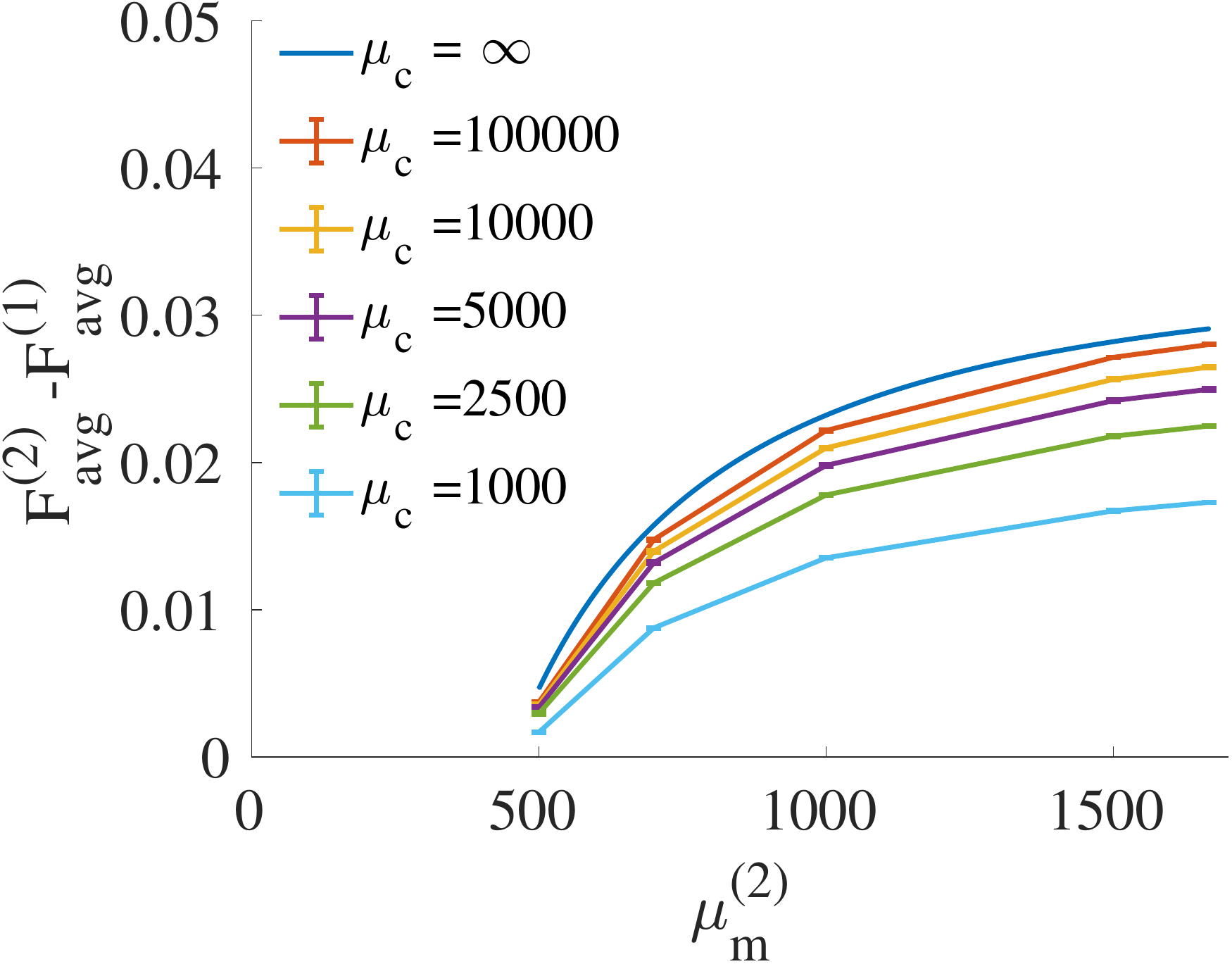}}
\vspace{-2mm}
\caption{Average gate fidelity differences in two entanglement request arrival and generation rate regimes. For all simulations, $T_1^{(1)}=T_1^{(2)}=0.00286$s and $T_2^{(1)}=T_2^{(2)}=0.001$s; $\mu_m^{(1)}=1667$Hz, $\lambda_c = 150$ and $\lambda_m = 1000$.}
\label{fig:sim_comp_fidelity_low_T1T2}
\end{figure}
Our goal in this section is to study the average gate and entanglement fidelities for the two architectures in a variety of settings in order to gain an understanding of regimes that are most suitable to each. We will also explore differences in manufacturing quality, and examine cases where it is preferable to use an SD design of better quality than a DD design with poorer quality. This question is especially relevant when cost-effectiveness is an important factor, as the DD design is expected to be the more expensive option (when comparing to an SD design of identical quality).

We use MATLAB to simulate the architectures and obtain the waiting time distributions for computational and entanglement jobs. We then use NetSquid \cite{coopmans2020netsquid} to simulate the storage of qubits according to the obtained waiting time distributions.
NetSquid is a discrete-event network simulator for quantum information; it provides a hardware-validated model of the NV center in diamond platform, and we use this model to evaluate the gate and entanglement fidelities.

Our analytical results apply to the case when computational jobs have negligible processing time ($\mu_c=\infty$); thus, simulating the case where $\mu_c < \infty$ provides additional insight. When $\mu_c<\infty$, computational jobs may no longer be processed instantaneously. Thus, we require an additional rule in handling multiple computational requests in the queue. Motivated by the fact that entanglement generation with remote nodes is the most time-consuming operation for our architectures, we endow state transfer (moving) jobs non-preemptive priority over computational jobs, \ie, when a moving job arrives while a computation is in progress, the former begins processing immediately after the completion of the computational request, even if other computational jobs were already in the queue prior to its arrival. In all other cases, jobs processed according to a FIFO policy.

For the following discussion, all rates are in terms of ($\#$ arrivals)/sec and ($\#$ jobs processed)/sec, unless otherwise specified. In all simulations, each run lasts for $10^5$ s, and each data point is an average of five runs -- a number that ensures sufficiently small error bars. We next motivate some of the parameter values used in the remainder of this section, many of which are inspired by the NV center in diamond platform. For state transfers within the SD architecture, we fix $\mu_m^{(1)}$ at $1667$Hz, since \cf Figure 4 in \cite{pompili2021realization}, a local swap to memory consumes $600\mu s$. State transfers within the DD architecture depend on the exact implementation of the transfer procedure across the inter-device interface. For the DD architecture, we often fix $\mu_m^{(2)}=700$Hz, as we expect the state transfer rate in this architecture to be between a third and a half of that of the SD architecture. Note that, according to \cite{pompili2021realization} and \cite{ruf2021quantum}, for the NV, this value is rather optimistic, since \eg, a Bell-state measurement -- an operation that is part of the state transfer gate sequence for the DD architecture --  alone consumes 1ms. Next, when exploring different values for the computation rate, we consider the variation in not only gate duration, but also the possibility of more time-consuming atomic gate sequences that must not be interrupted by any other operations; using \cite{linkLayerPaper} as a guide, we set $\mu_c$ to values in the range $[10^3,10^5]$. For $T_1$ and $T_2$ values we use \cite{linkLayerPaper} as a guide.
Finally, when choosing parameter values for the request arrival rates, we ensure that both the SD and DD systems are stable in the number of outstanding entanglement requests. 
\subsubsection*{Effects of Device Memory Lifetimes on Gate Fidelity}
Recall from our analytical evaluation that the DD architecture achieves higher gate fidelity than the SD architecture when both devices have the same $T_1$ and $T_2$ parameters characterizing their memory lifetime. This phenomenon can be observed in Figure \ref{fig:sim_comp_fidelity_low_T1T2} for two different entanglement generation regimes: $\mu_e=10$ corresponds to a quantum network setting, where the remote node(s) is(are) distant, while $\mu_e=500$ represents closely-located quantum nodes, \eg, as may be found in a quantum computing cluster. Observe that for the latter, the fidelity differences are less pronounced. This can be explained by the fact that faster entanglement generation rates are less detrimental to the SD architecture's gate fidelity than slower ones, as computational jobs wait less before being serviced.
\begin{figure}[t]
\centering
\subfloat[$\mu_e=10$, $\lambda_e=1$]{\includegraphics[width=0.4\textwidth]{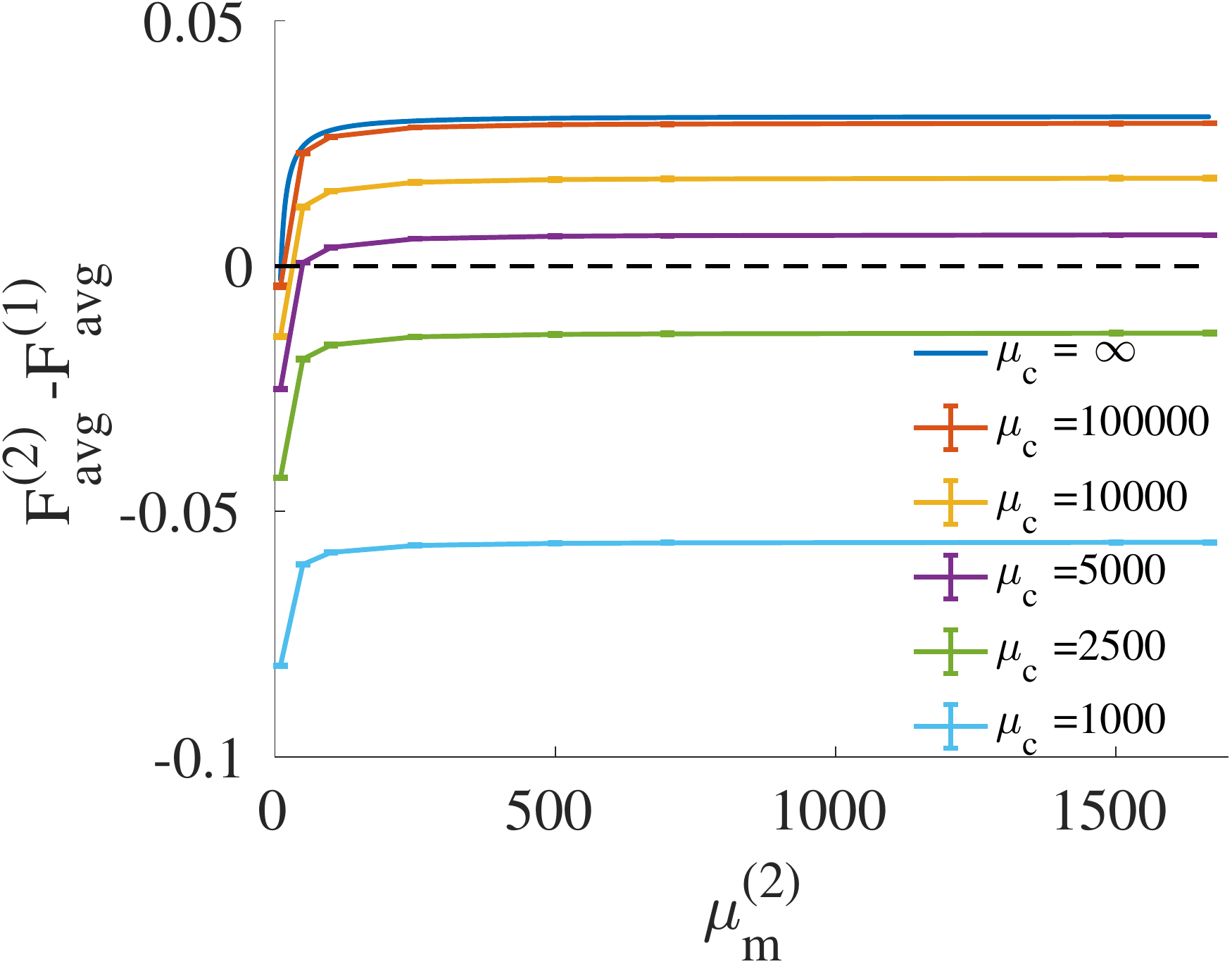}}\qquad\qquad
\subfloat[$\mu_e=500$, $\lambda_e=50$]{\includegraphics[width=0.4\textwidth]{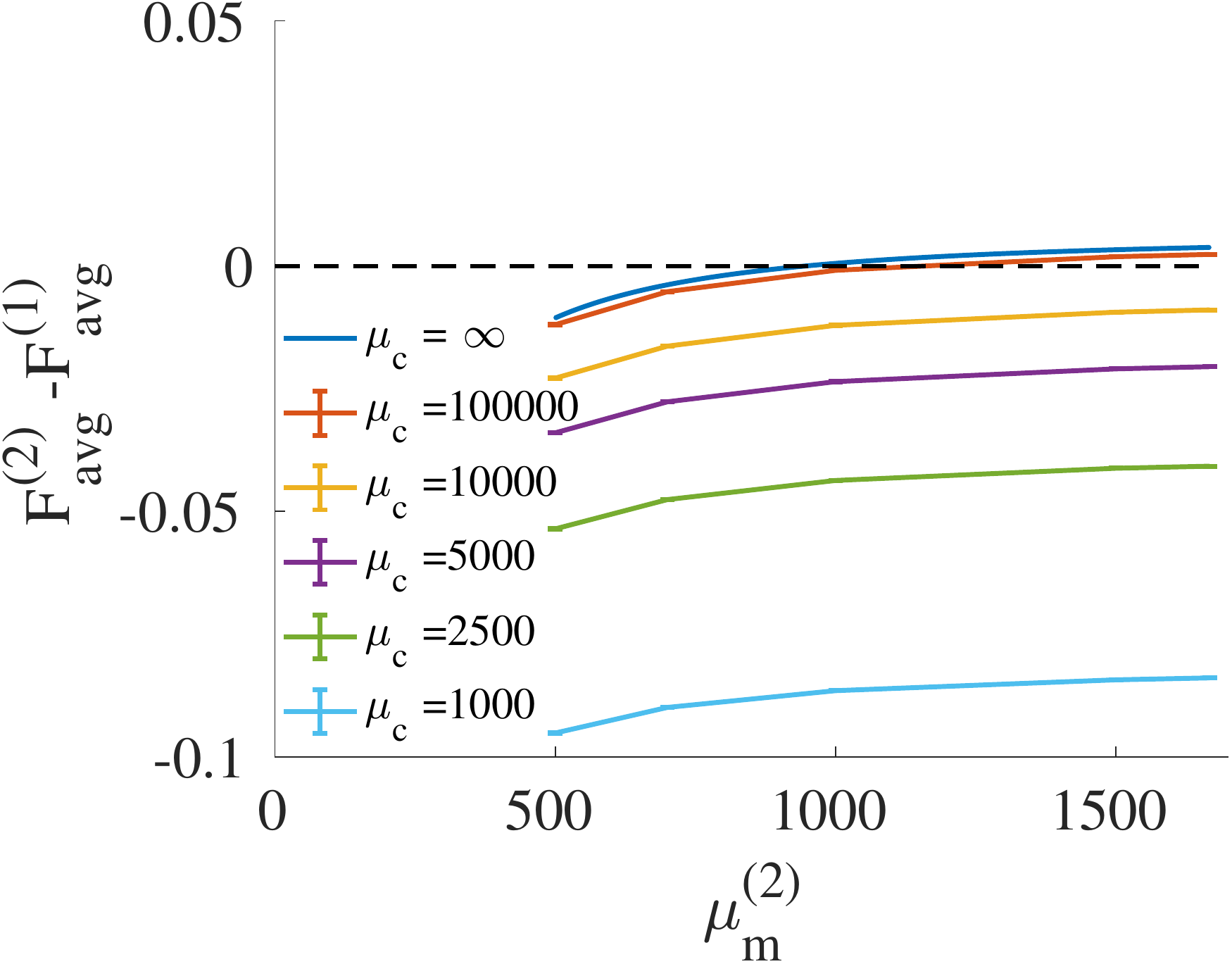}}
\caption{Average gate fidelity differences in two entanglement request arrival and generation rate regimes. For all simulations, $T_1^{(1)}=10$s, $T_2^{(1)}=0.01$s, $T_1^{(2)}=2$s, and $T_2^{(2)}=0.002$s; $\mu_m^{(1)}=1667$Hz, $\lambda_c = 150$ and $\lambda_m = 1000$.}
\label{fig:sim_comp_fidelity_high_T1T2}
\end{figure}
\begin{figure}[t]
\centering
\begin{minipage}[t]{0.47\textwidth}
\centering
\includegraphics[width=1.05\textwidth]{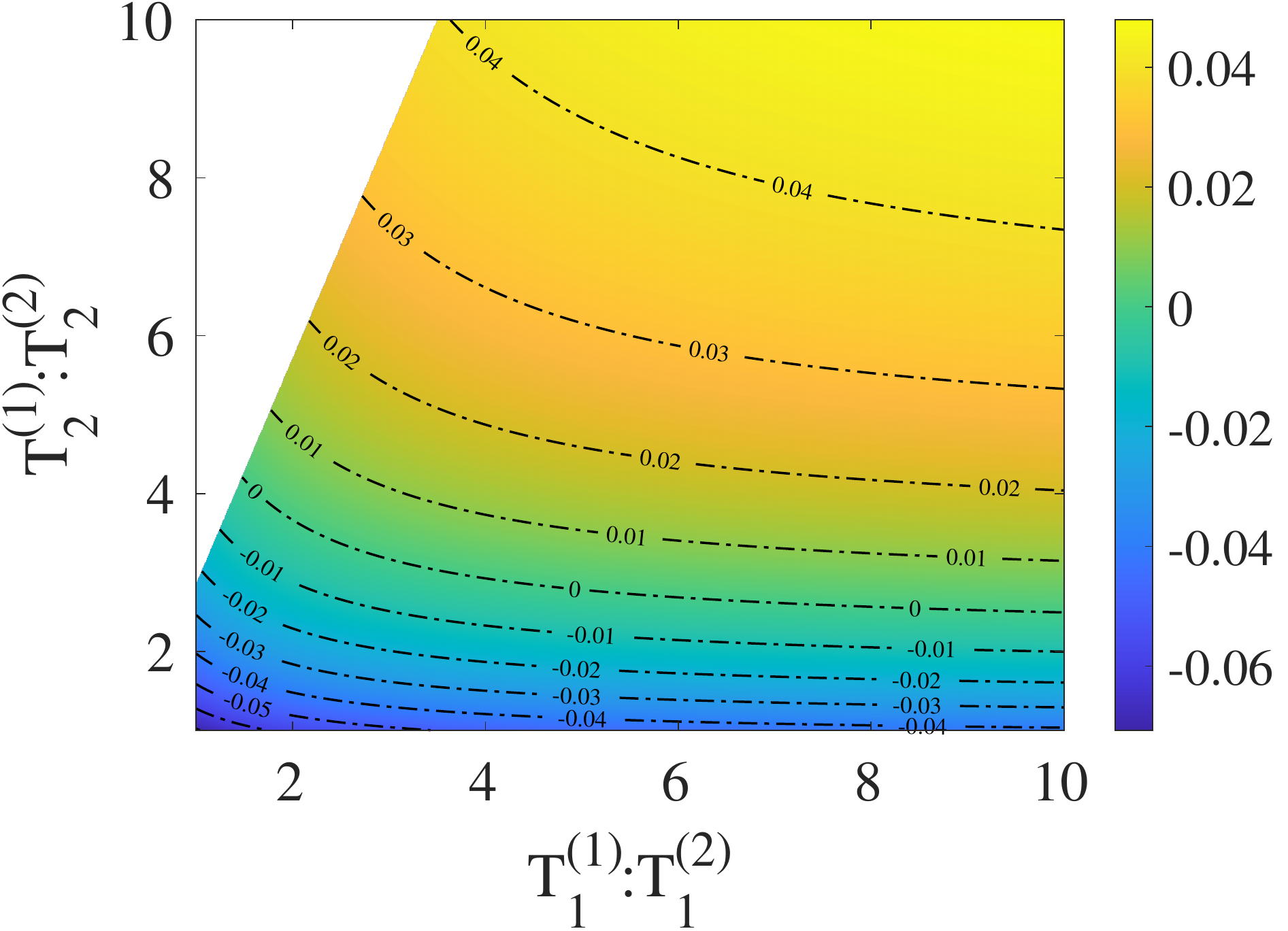}
\captionof{figure}{Differences in average gate fidelity for the SD and DD architectures. For the latter, $T_1^{(2)}$ and $T_2^{(2)}$ are fixed at 0.00286s and 0.001s, respectively, while the ratios of the $T_1$ and $T_2$ times are varied. Here, $\mu_c=\infty$, $\lambda_e=225$, $\mu_e=500$, $\mu_m^{(1)}=1667$Hz, $\mu_m^{(2)}=700$Hz, and $\lambda_m=1000$ to ensure stability.}
\label{fig:comp_fid_v_T1T2}
\end{minipage}\hfill
\begin{minipage}[t]{0.45\textwidth}
\centering
\includegraphics[width=\textwidth]{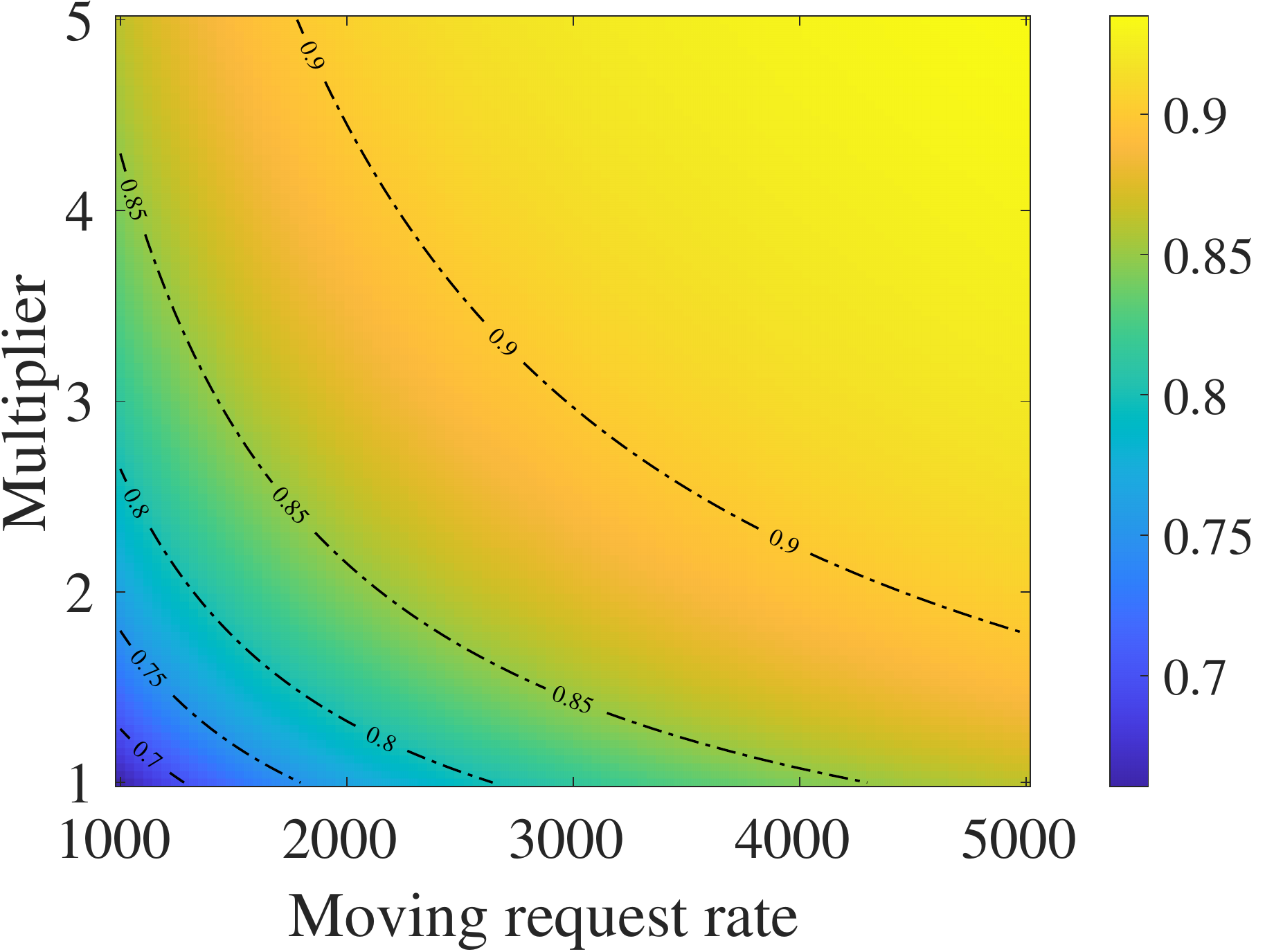}
\captionof{figure}{Post-move fidelity for the SD architecture, as a function of the moving request rate $\lambda_m$. The characteristic memory times are varied at a constant ratio $T_1^{(1)}:T_2^{(1)}$, with initial values (\ie, when the multiplier is set to one) of $0.00286$s and $0.001$s for $T_1^{(1)}$ and $T_2^{(1)}$, respectively.}
\label{fig:arch1_pm_fidelity_v_T1T2}
\end{minipage}
\end{figure}
Fabrication of networked quantum hardware is a complex task, and achieving comparable memory lifetimes between the SD and DD architectures may prove to be difficult.  Interfacing the two processors that make up the DD architecture may introduce additional sources of noise or complicate the process of properly shielding qubits in order to maintain adequate memory lifetimes. As a result, it is possible that a  high-quality SD design is both the more economical and functional choice, compared to a poorer-quality DD design. A potential instance of this is presented in Figure \ref{fig:sim_comp_fidelity_high_T1T2}, where the memory lifetimes of SD are five times those of the DD design. The advantages of the SD design are especially notable in the higher entanglement generation rate regime.

In Figure \ref{fig:comp_fid_v_T1T2}, we explore the effect of manufacturing differences on the average gate fidelity further, focusing only on the analytical case ($\mu_c=\infty$). Here, we assume that $T_1^{(1)} > T_1^{(2)}$ and $T_2^{(1)} > T_2^{(2)}$. Each data point in the figure represents the difference between the SD and DD average gate fidelities for given $T_{1/2}^{(1/2)}$ values, with positive values representing cases where the SD performs better than the DD architecture. Note that not all combinations of $T_1$ and $T_2$ values are possible, as per the usual requirement that $T_1> T_2$ for each architecture. From the figure, we observe that the $T_2$ time plays a significant role in improving gate fidelity: note in particular that for lower values of $T_2^{(1)}$, the SD architecture does not outperform the DD architecture, even for very high values of $T_1^{(1)}$.
\subsubsection*{Effects of Processing Rates on Gate Fidelity}
Processing rates impact the amount of time needed to complete requests and consequently result in longer waiting times. 
Our analytical evaluation assumed that $\mu_c=\infty$ so that computational requests do not interfere with entanglement and moving requests.
We use the computational job processing rate, $\mu_c$, to capture the behavior of different applications, where smaller processing rates may correspond to computation requests performing a sequence of atomic gates or 
computation requiring ``implicit'' state transfers (see Appendix \ref{app:NVoverview}).
In both Figures \ref{fig:sim_comp_fidelity_low_T1T2} and \ref{fig:sim_comp_fidelity_high_T1T2}, we observe that more time-consuming computations are more hospitable to the SD architecture's gate fidelity. This phenomenon arises from the non-preemptive priority of moving jobs, which interrupt computation for longer periods in the DD architecture.
To see how the individual average gate fidelities evolve as a function of $\mu_c$, see Appendix \ref{app:indivFidelityPlots}. 
\subsubsection*{Effects of Arrival Rates on Gate Fidelity}
The computation, entanglement, and moving request arrival rates impact the number of jobs that are issued to the system and consequently the probability that jobs block one another when non-zero processing time is assumed.  Since no jobs can be processed in parallel on the SD architecture, this leads to an overall increase in waiting time for all requests.  In contrast, computation jobs in the DD architecture are only blocked by moving jobs, while entanglement jobs may be processed in parallel to computation jobs.
When considering the results in Figures \ref{fig:sim_comp_fidelity_low_T1T2} and \ref{fig:sim_comp_fidelity_high_T1T2}, one key observation is that both the entanglement request and moving request arrival rates impact the the gate fidelity of computation jobs on the DD architecture.  By increasing the arrival rate of entanglement requests, additional moving requests are submitted which leads to a decrease in the gate fidelity of computation requests (as further evidenced by Figures \ref{fig:sep_sim_comp_fidelity_low_T1T2} and \ref{fig:sep_sim_comp_fidelity_high_T1T2} in Appendix \ref{app:indivFidelityPlots}).  In contrast, we observe an increase in the gate fidelity of the SD architecture here as the entanglement generation rate also increases, leading to a decrease in the amount of time that queued computation requests are blocked from processing.
\begin{figure}[t]
\centering
\subfloat[DD architecture pre- and post-move fidelities.]{\includegraphics[width=0.45\textwidth]{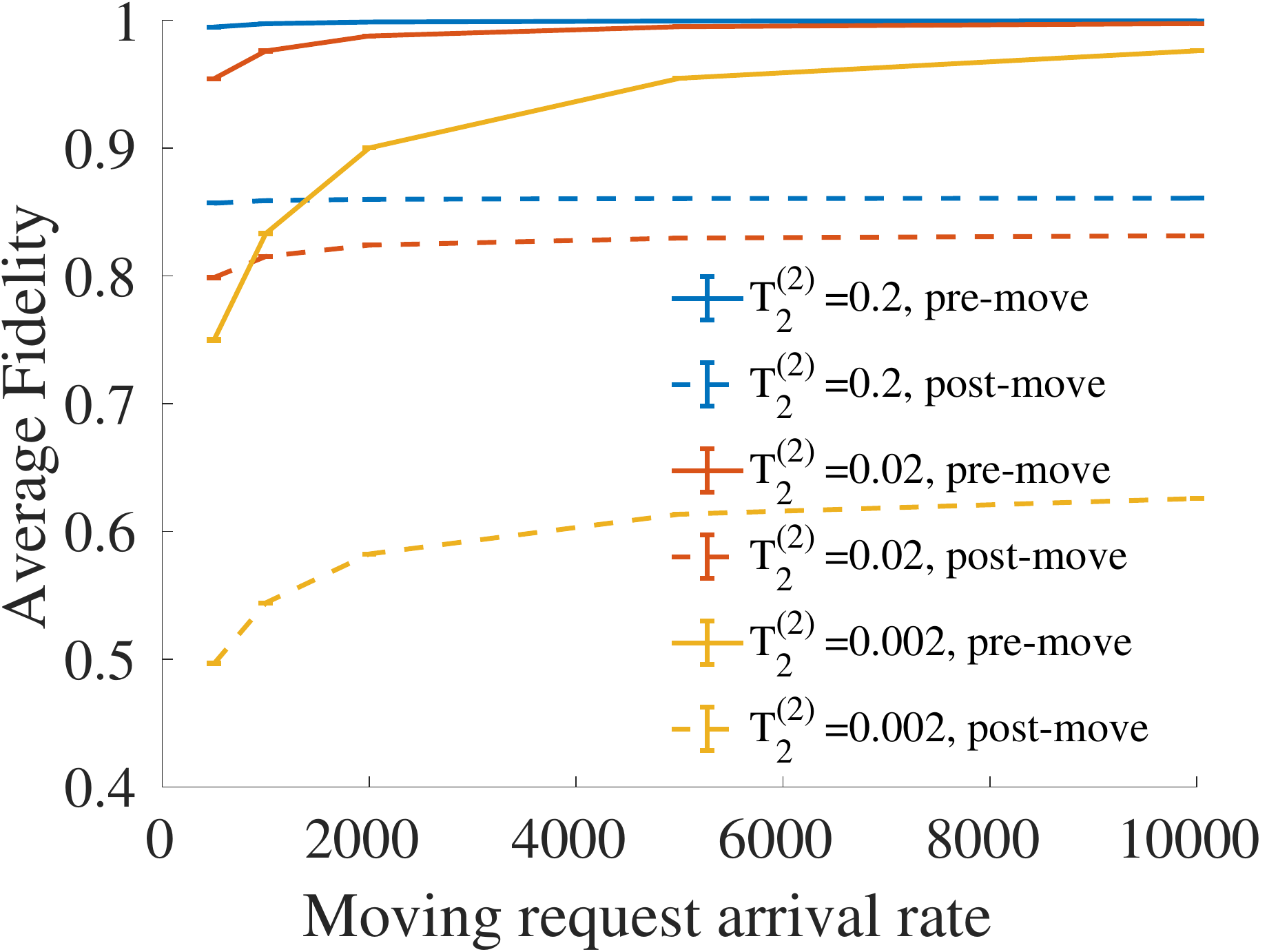}\label{fig:arch12_sim_pm_ent_fidelity_DD}}
\qquad\quad
\subfloat[SD and DD architecture post-move fidelities. ]{\includegraphics[width=0.45\textwidth]{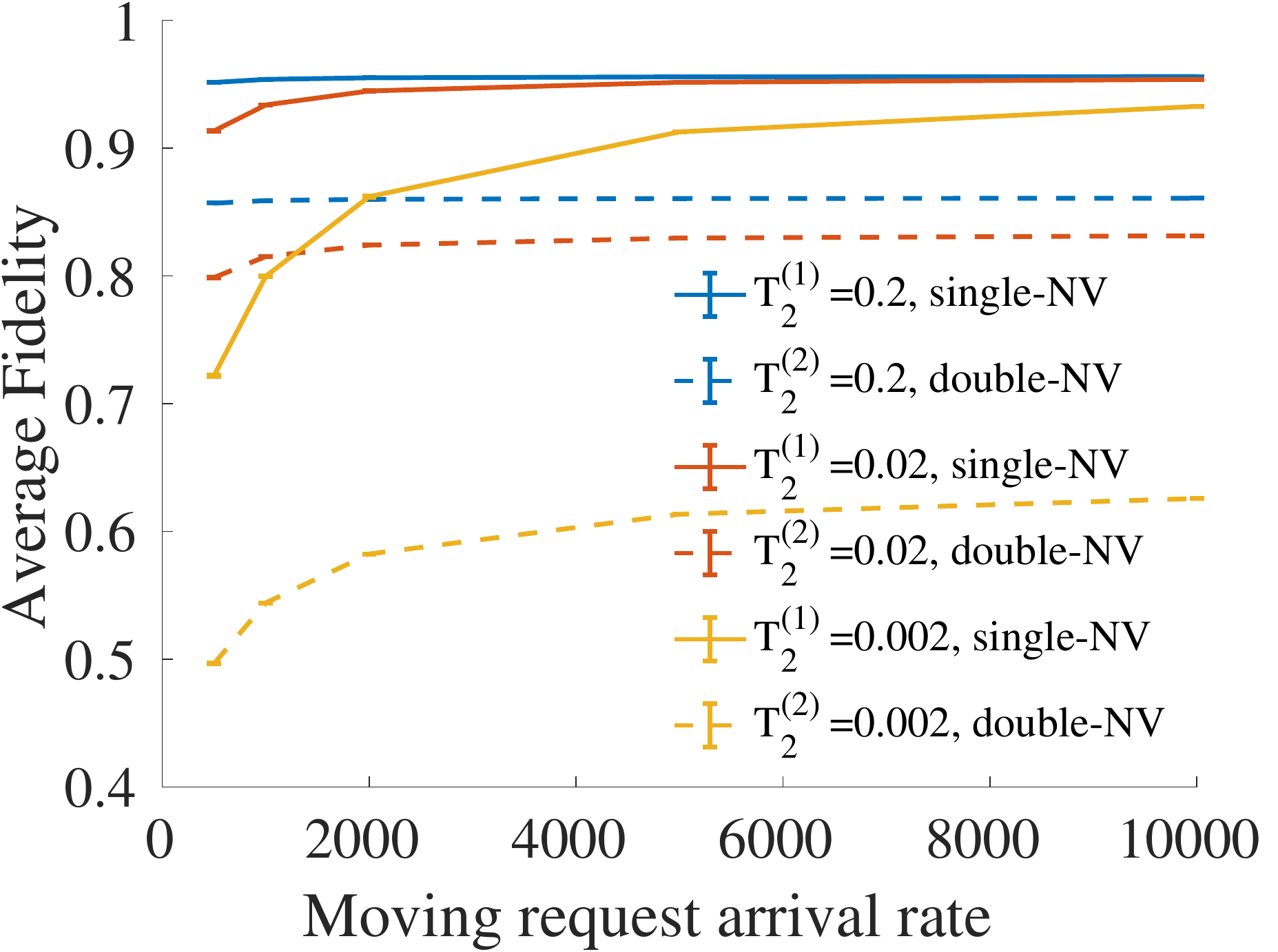}\label{fig:arch12_sim_pm_ent_fidelity_comp}}
\caption{Pre- and post-move fidelities for the DD architecture, $(a)$, and post-move fidelities for the two architectures, $(b)$. Here, $T_1^{(1/2)}=2$s, $\lambda_e=1$, $\mu_e=10$, $\lambda_c = 150$, $\mu_c=10^5$,  $\mu_m^{(1)}=1667$Hz, and $\mu_m^{(2)}=700$Hz.}
\label{fig:arch12_sim_pm_ent_fidelity}
\end{figure}
\subsubsection*{Effects of Moving Entanglement}
We now examine the effects of waiting time on the pre- and post-move entanglement fidelity in both architectures when they are realized with the NV center in diamond platform. Recall from Section \ref{sec:quantum_bg} that the SD and DD architectures process moving requests in different ways. The SD architecture performs a sequence of local gates whereas the DD architecture must execute a network operation to transfer entanglement from the networking to the computing device (see Appendix \ref{sec:nv_state_transfer} for further details).  Since quantum gates in both architectures are imperfect, it is important to highlight the differences in entanglement fidelity once a moving request has been processed. The \emph{pre-move} entanglement fidelity is computed immediately before the transfer request is processed. \emph{I.e.}, this measure incorporates the effects of waiting time both from needing to wait for the arrival of the transfer request, as well as the possible extra waiting time due an in-progress computation (recall that transfer requests have non-preemptive priority over computational jobs). The \emph{post-move} entanglement fidelity is computed immediately after the entanglement has been moved from the networking to the computing component.

For the SD design, we are able to obtain the post-move fidelity in closed form (see Appendix \ref{app:postMoveEntFidSingleNV}). We present it in Figure \ref{fig:arch1_pm_fidelity_v_T1T2}, as a function of the moving request rate $\lambda_m$ and the memory lifetimes of the architecture. As a reference, basic QKD demands entanglement fidelity of at least 0.81 \cite{Gottesman_2003}. 
Figure \ref{fig:arch1_pm_fidelity_v_T1T2} shows that, as expected, the moving request arrival rate (\ie, the application's responsiveness to processing newly-generated entanglement) may be reduced if higher memory lifetimes are available. For the DD architecture, we obtain the entanglement fidelity via NetSquid. Figure \ref{fig:arch12_sim_pm_ent_fidelity_DD} presents its pre- and post-move fidelities for varying $T_2^{(2)}$ times, a parameter to which the post-move fidelity is highly sensitive. Indeed, we observe that the moving request arrival rate is not nearly as impactful to the fidelity as $T_2^{(2)}$. Figure \ref{fig:arch12_sim_pm_ent_fidelity_comp} presents a comparison of the two architectures' post-move fidelities; the SD design clearly outperforms the DD design for each of the $T_2$ values.

\begin{figure}[t]
\centering
\subfloat[$\mu_e = 10,~\lambda_e=1$]{\includegraphics[width=0.45\textwidth]{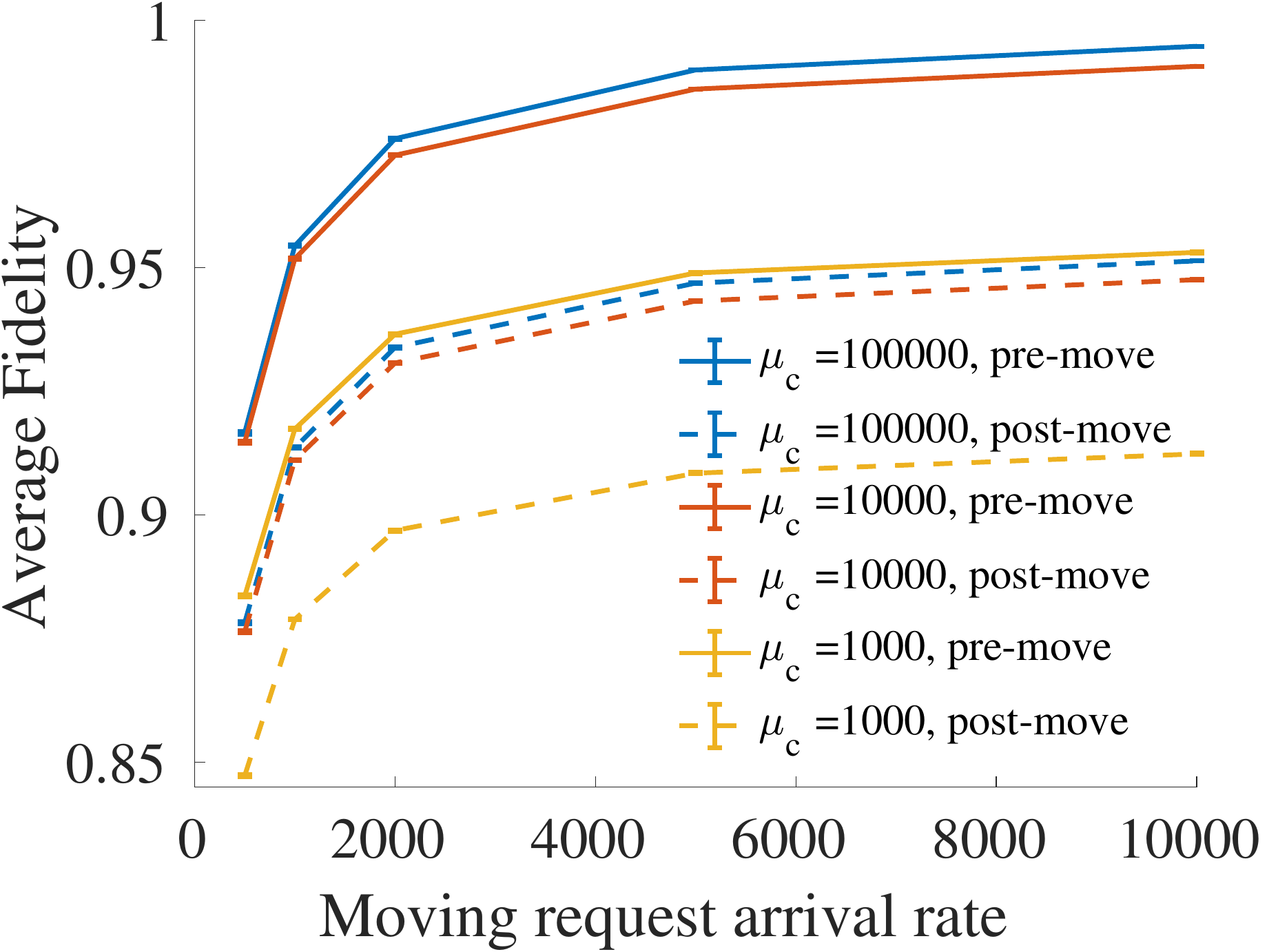}}\qquad\quad
\subfloat[$\mu_e = 1000,~\lambda_e=100$]{\includegraphics[width=0.45\textwidth]{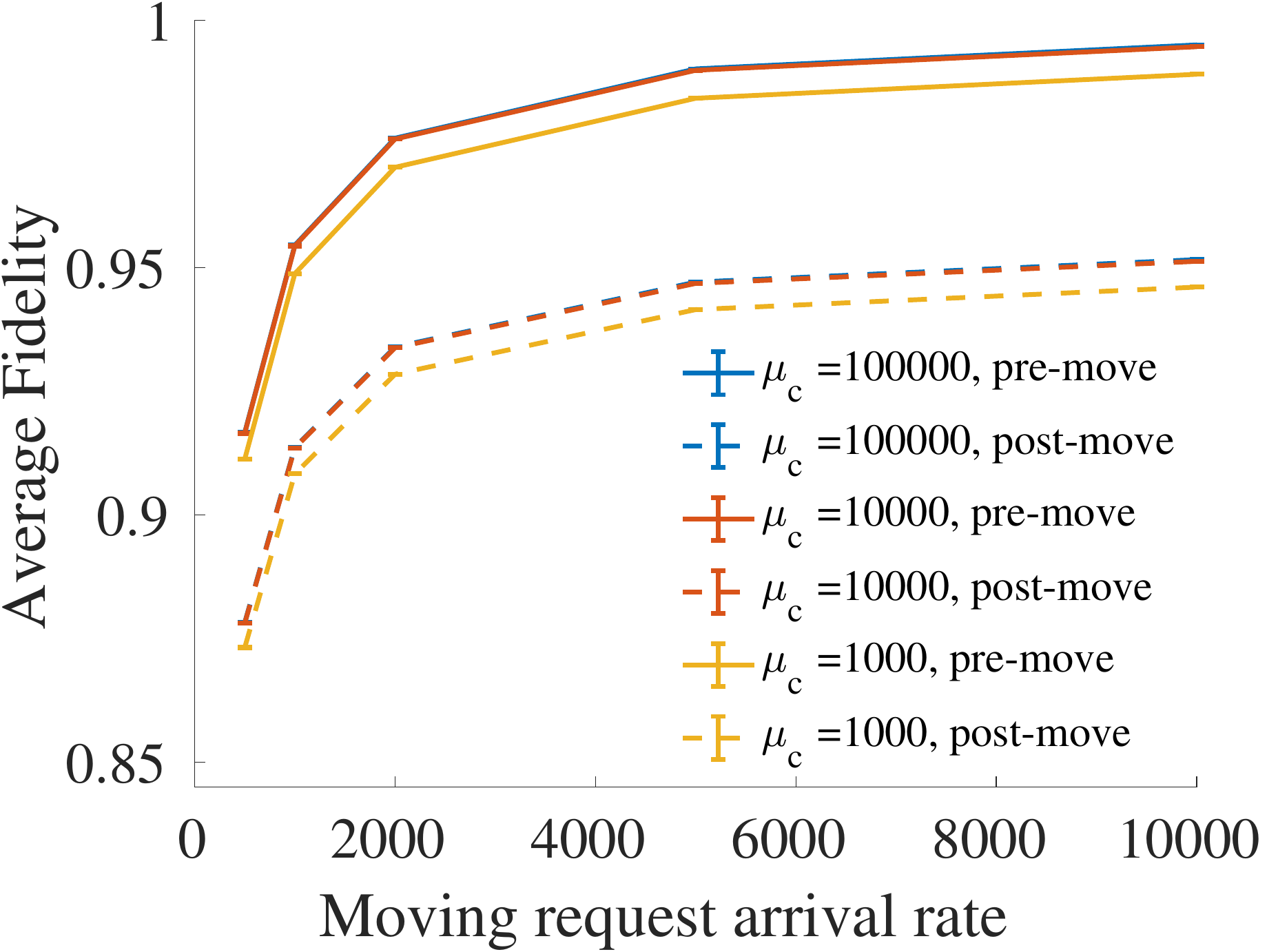}}
\caption{Pre-move and post-move fidelities for varying computational job processing rates in the SD architecture, as functions of the moving request arrival rate $\lambda_m$. Here, $T_1^{(1)}=10$s and $T_2^{(1)}=0.01$s; $\lambda_c=150$ and $\mu_m^{(1)}=1667$Hz for all SD simulations and we assume the qubits are initially in a perfect entangled state.}
\label{fig:sim_ent_fidelity_high_T1T2}
\end{figure}
Figure \ref{fig:sim_ent_fidelity_high_T1T2} shows the reduction between the pre-move and post-move entanglement fidelities for the SD architecture, in two entanglement generation regimes. It is immediately evident that in the high entanglement generation rate regime (\eg, one that may be representative of a distributed quantum cluster setting), the pre- and post-move fidelities fair far better than in the low entanglement generation rate regime (\eg, one that is more representative of a quantum network with distant nodes).
A reason for this is that faster processing of entanglement requests in the SD design frees up processing time for computation. Consequently, there are  fewer computation requests left in the queue, so that new entanglement requests can be processed more quickly as well.

\subsubsection*{Summary of findings} From our analysis and numerical observations, we find stark contrasts between the two architectures. On the one hand, when implemented with memories of identical quality, the DD design dominates in terms of gate fidelity. However, in a more practical scenario, wherein the DD design's more complex manufacturing would impair its memory lifetimes, the SD design can yield higher gate fidelities, and is more robust to longer computation times. Further, for present-day parameters, the SD design is more hospitable to the entanglement fidelity. The advantages of the SD design are especially evident in the high entanglement generation rate regime. We thus conclude that 
 the DD design is more suitable for 
 settings such as long-distance quantum communication, with lower entanglement generation rates and lighter computational demands. In contrast, the SD design is better suited for settings such as a distributed quantum computing cluster where high entanglement generation rates can be achieved and longer computations must be performed.

\section{Conclusion}
\label{sec:conclusion}
Quantum distributed applications impose quality constraints on the quantum states that they consume. When such applications are executed on architectures with physical limitations, such as imperfect gates or a limited amount of parallelism, resource contention can significantly impact performance, as some quantum states may be forced to wait in storage while others are being processed.
In this work, we studied the effects of waiting times on the gate and entanglement fidelities for two distributed quantum architectures. We accomplished this by deriving formulas for average fidelity as a function of the waiting time distribution for a quantum state awaiting processing, as well as a noise model that governs quantum state evolution during storage. We obtained the waiting time distributions from the analysis of a Markov chain that models both of the quantum architectures in a regime where computation consumes a negligible amount of time; we later relaxed this assumption to study the effects of more time-consuming computation via simulation. We discovered that certain architecture implementations are more suitable for environments that are computation-heavy, while others are suitable for entanglement-heavy applications. Our average fidelity formulas are applicable in scenarios beyond those studied in this work, and may serve as a useful tool for performance evaluation experts.

Several extensions of our problem formulation are possible. First, we examined only two possible architecture implementations in this manuscript, but other realizations of distributed quantum architectures can be proposed. For instance, in the DD architecture, one could equip both of the devices with an interface to the outside world, so that both devices can perform local computation as well as remote entanglement generation. It is not entirely obvious what advantages one would gain in such a setup, and much like we have observed in the current work, the performance of such a system will depend on entanglement generation rates and the quality of the interface between the two devices (for transporting qubits from one to the other), as well as the application type (\eg, computation-heavy vs. entanglement-heavy).
Also, in the current work we assumed that each architecture has a single link to the outside world; an extension of the problem would be to consider multiple links, each associated with a different entanglement generation rate.
Second, even with the two architectures examined in the current work, we have not studied all possible use cases. For instance, in the SD design, one could take further advantage of processor idle time by allowing computation to occur during entanglement generation, while the device is awaiting a heralding signal from the link. This would require the state of the (not yet entangled) qubit to be moved to a storage qubit, and then back again to the communication qubit; thus, for a rigorous analysis, one would have to account for how these state transfer operations would reflect on the final entanglement fidelity. Finally, from a modeling perspective, one could relax several of our assumptions, \eg, that computational jobs consume zero time -- something that we have so far only explored via simulation.
\begin{acks}
This work was supported in part by the NWO ZK QSC Ada Lovelace Fellowship. The authors thank Filip Rozp{\k{e}}dek, Guus Avis, Francisco Ferreira da Silva, and David Maier for useful discussions and careful reading of an earlier version of the manuscript.
\end{acks}

\bibliographystyle{ACM-Reference-Format}
\bibliography{references}

\appendix
\section{CTMC Analysis}
\subsection{Rate Matrix}
\label{ctmc_analysis_app}
Define the following useful variables:
\begin{align}
\alpha \coloneqq \lambda_e+\mu_e, \quad \beta\coloneqq \lambda_e+\lambda_m, \quad\text{and}\quad \gamma\coloneqq \lambda_e+\mu_m.
\end{align}
The infinitesimal generator of the CTMC in Figure \ref{fig:ctmcHypoK1} is given by
\begin{align*}
Q &= \begin{bmatrix}
B_{00} & B_{01} & \vect{0}& \cdots\\
B_{10} & A_1 & A_2 & \vect{0}& \cdots\\
\vect{0} & A_0 & A_1 & A_2 & \vect{0}& \cdots\\
\vect{0} & \vect{0} & A_0 & A_1 & A_2 & \vect{0}& \cdots\\
\vdots & \vdots & \quad\ddots & \quad\ddots & \quad\ddots & \quad\ddots
\end{bmatrix},
\end{align*}
where 
\begin{align*}
B_{00} &= -\lambda_e,\quad B_{01} = \begin{bmatrix}\lambda_e & 0 & 0\end{bmatrix},\quad B_{10} = \begin{bmatrix} 0 \\ 0\\ \mu_m\end{bmatrix},\\\\
A_0 &= \begin{bmatrix}0 & 0 & 0\\
0 & 0 & 0\\
\mu_m & 0 & 0\end{bmatrix},
\quad A_1=\begin{bmatrix}-\alpha & \mu_e & 0\\
0 & -\beta & \lambda_m\\
0 & 0 & -\gamma
\end{bmatrix},
\quad\text{and}\quad A_2 = \begin{bmatrix}
\lambda_e & 0 & 0\\
0 & \lambda_e & 0\\
0 & 0 & \lambda_e
\end{bmatrix},
\end{align*}
and the $\vect{0}$'s are vectors or matrices of appropriate dimensions. The ergodicity condition for Markov chains whose generators have tridiagonal block structure are well-studied in literature, see, \eg, \cite{neuts1994matrix}; we derive it here for completeness. The QBD process driven by $Q$ is ergodic if and only if 
\begin{align}
\vect{\omega}A_2\vect{e} < \vect{\omega}A_0\vect{e},
\label{eq:ergodicitycond}
\end{align}
where $\vect{\omega}$ is the equilibrium distribution of the generator $A_0+A_1+A_2$ and $\vect{e}$ is a vector of all ones. To find $\vect{\omega}$, we use the relation
\begin{align*}
\vect{\omega}(A_0+A_1+A_2) = \vect{0},
\end{align*}
or
\begin{align*}
\begin{bmatrix}\omega_1 & \omega_2 & \omega_3\end{bmatrix}\begin{bmatrix}-\mu_e & \mu_e & 0\\
0 & -\lambda_m & \lambda_m\\
\mu_m & 0 & -\mu_m
\end{bmatrix} = \begin{bmatrix}0 & 0 &0\end{bmatrix},
\end{align*}
which, along with the normalizing condition on $\omega_i$, yield
\begin{align*}
\omega_1 = \frac{\lambda_m\mu_m }{\lambda_m\mu_m+\lambda_m\mu_e+\mu_e\mu_m},
\quad \omega_2 = \frac{\mu_e\mu_m }{\lambda_m\mu_m+\lambda_m\mu_e+\mu_e\mu_m},\quad\text{and}\quad
\omega_3 = \frac{\lambda_m\mu_e}{\lambda_m\mu_m+\lambda_m\mu_e+\mu_e\mu_m}.
\end{align*}
Thus, after defining $D \coloneqq \lambda_m\mu_m+\lambda_m\mu_e+\mu_e\mu_m$ and with the assumption that $D>0$ since all rates are positive, (\ref{eq:ergodicitycond}) becomes
\begin{align*}
\frac{1}{D}\begin{bmatrix}\lambda_m\mu_m & \mu_e\mu_m & \lambda_m\mu_e\end{bmatrix}\begin{bmatrix}
\lambda_e & 0 & 0\\
0 & \lambda_e & 0\\
0 & 0 & \lambda_e
\end{bmatrix}\begin{bmatrix}1\\1\\1\end{bmatrix} &<
\frac{1}{D}\begin{bmatrix}\lambda_m\mu_m & \mu_e\mu_m & \lambda_m\mu_e\end{bmatrix}\begin{bmatrix}0 & 0 & 0\\
0 & 0 & 0\\
\mu_m & 0 & 0\end{bmatrix}\begin{bmatrix}1\\1\\1\end{bmatrix},\\\\
\lambda_e\begin{bmatrix}\lambda_m\mu_m & \mu_e\mu_m & \lambda_m\mu_e\end{bmatrix}\begin{bmatrix}1\\1\\1\end{bmatrix} &<
\begin{bmatrix}\lambda_m\mu_m & \mu_e\mu_m & \lambda_m\mu_e\end{bmatrix}\begin{bmatrix}0\\0\\\mu_m\end{bmatrix},\\
\lambda_e(\lambda_m\mu_m + \mu_e\mu_m + \lambda_m\mu_e) &<
\lambda_m\mu_e\mu_m,
\end{align*}
or, written another way,
\begin{align}
\frac{1}{\lambda_e} > \frac{1}{\mu_e}+\frac{1}{\lambda_m}+\frac{1}{\mu_m}.
\label{eq:stabcondHypo}
\end{align}
Intuitively, (\ref{eq:stabcondHypo}) indicates that for ergodicity, the average time between entanglement request arrivals must exceed the average processing times summed over all three stages (entanglement generation, waiting for a moving request to arrive, and performing the moving operation).

Henceforth, assume that (\ref{eq:stabcondHypo}) is satisfied. Next, we obtain the rate matrix of the generator. Note that $A_0$ is of rank 1; we may rewrite it as follows:
\begin{align*}
A_0 = \begin{bmatrix}0 \\ 0\\ \mu_m\end{bmatrix}\begin{bmatrix}1 & 0 & 0\end{bmatrix}.
\end{align*}
By the results in \cite{latouche1999introduction}, this means that the rate matrix may be computed explicitly; it is given by
\begin{align}
R &= -A_2(A_1+A_2\vect{e}\begin{bmatrix}1 & 0 & 0\end{bmatrix})^{-1}.
\end{align}
After some algebra, we obtain
\begin{align}
R = \frac{\lambda_e}{\lambda_m\mu_e\mu_m}\begin{bmatrix}
\beta\gamma & \gamma\mu_e & \mu_e\lambda_m\\
\lambda_e(\gamma+\lambda_m) & \gamma\mu_e & \mu_e\lambda_m\\
\lambda_e\beta & \lambda_e\mu_e & \mu_e\lambda_m
\end{bmatrix},
\end{align}
which matches the explicit rate matrix computation in \cite{marin2011explicit} for the $M/HYPO_K/1$ queue when $K=3$.

\subsection{Probability that a Computational Job Must Wait}
\label{app:stationary}
Our goal here is to compute $\sum\limits_{N=1}^{\infty}\pi_{N/1}$ and $\sum\limits_{N=1}^{\infty}\pi_{N/3}$.
To derive these quantities, we use the global balance principle to ``cut'' the chain  three different ways (the first cut isolates the states $N/3$, the second isolates the states $0$ and $N/1$, and the third isolates the states $N/2$), obtaining the following balance equations (below, $\pi_{N/2}$ is the stationary probability of state $N/2$):
\begin{align}
\mu_m\sum\limits_{N=1}^{\infty}\pi_{N/3} &= \lambda_m\sum\limits_{N=1}^{\infty}\pi_{N/2},\label{eq:globBal1}\\
\lambda_m\sum\limits_{N=1}^{\infty}\pi_{N/2} &= \mu_e\sum\limits_{N=1}^{\infty}\pi_{N/1},\label{eq:globBal2}\\
\mu_e\sum\limits_{N=1}^{\infty}\pi_{N/1} &= \mu_m\sum\limits_{N=1}^{\infty}\pi_{N/3}.\label{eq:globBal3}
\end{align}
We can rewrite (\ref{eq:globBal1}) as follows:
\begin{align}
\mu_m\sum\limits_{N=1}^{\infty}\pi_{N/3} &= \lambda_m\left(1-\sum\limits_{N=1}^{\infty}\pi_{N/1}-\sum\limits_{N=1}^{\infty}\pi_{N/3}-\pi_0\right),
\label{eq:pil3_1}
\end{align}
where $\pi_0$ is the stationary probability of state $0$,
and using (\ref{eq:globBal3}), (\ref{eq:pil3_1}) becomes
\begin{align}
\mu_m\sum\limits_{N=1}^{\infty}\pi_{N/3} &= \lambda_m\left(1-\frac{\mu_m}{\mu_e}\sum\limits_{N=1}^{\infty}\pi_{N/3}-\sum\limits_{N=1}^{\infty}\pi_{N/3}-\pi_0\right),\\
\left(\mu_m+\frac{\lambda_m\mu_m}{\mu_e}+\lambda_m\right)\sum\limits_{N=1}^{\infty}\pi_{N/3} &= \lambda_m\left(1-\pi_0\right),\\
\sum\limits_{N=1}^{\infty}\pi_{N/3} &= \frac{\lambda_m\left(1-\pi_0\right)}{\mu_m+\frac{\lambda_m\mu_m}{\mu_e}+\lambda_m}
= \frac{\lambda_m\left(1-\pi_0\right)\mu_e}{\mu_m\mu_e+\lambda_m\mu_m+\lambda_m\mu_e}.
\label{eq:pil3sum}
\end{align}
To obtain $\pi_0$, we use the balance equations of the QBD process along with the normalizing condition:
\begin{align}
\begin{bmatrix}\pi_0 & \vect{\pi_1}\end{bmatrix}\begin{bmatrix}B_{00} & B_{01}\\
B_{10} & A_1+RA_0\end{bmatrix} &= \begin{bmatrix}0 & \vect{0}\end{bmatrix},\label{eq:pi01_hypo}\\
\pi_0 + \vect{\pi_1}(I-R)^{-1}\vect{e} & = 1,\label{eq:pi02_hypo}
\end{align}
where $\vect{\pi_1}\equiv \begin{bmatrix}\pi_{1/1} & \pi_{1/2} & \pi_{1/3}\end{bmatrix}$;
$B_{00}$, $B_{01}$, $B_{10}$, $A_1$, and $A_0$ are blocks of the CTMC generator matrix defined in Appendix \ref{ctmc_analysis_app}, $R$ is the rate matrix of the QBD process derived in Appendix \ref{ctmc_analysis_app}, $\vect{e}$ is a vector of all ones, and $I$ is an identity matrix of the same dimensions as $R$.

From (\ref{eq:pi01_hypo}) and (\ref{eq:pi02_hypo}), we obtain
\begin{align}
\pi_0 &= \frac{\Delta}{\lambda_m\mu_e\mu_m},\quad 
\pi_{1/1} = \frac{\lambda_e(\lambda_e+\lambda_m)(\lambda_e+\mu_m)\Delta}{(\lambda_m\mu_e\mu_m)^2},\quad
\pi_{1/2} = \frac{\lambda_e(\lambda_e+\mu_m)\Delta}{\lambda_m^2\mu_e\mu_m^2},\quad
\pi_{1/3} = \frac{\lambda_e\Delta}{\lambda_m\mu_e\mu_m^2},
\label{eq:statdistrpi1}
\end{align}
where $\Delta \equiv (\lambda_m\mu_e\mu_m - \lambda_e\lambda_m\mu_e - \lambda_e\lambda_m\mu_m - \lambda_e\mu_e\mu_m )$. Note that $\Delta>0$ follows directly from the ergodicity of the chain.
Using (\ref{eq:pil3sum}), (\ref{eq:statdistrpi1}) and the definition of $\Delta$, we obtain
\begin{align}
\sum\limits_{N=1}^{\infty}\pi_{N/3} &= \frac{\lambda_m\left(1-\frac{\lambda_m\mu_e\mu_m - \lambda_e\lambda_m\mu_e - \lambda_e\lambda_m\mu_m - \lambda_e\mu_e\mu_m}{\lambda_m\mu_e\mu_m}\right)\mu_e}{\mu_m\mu_e+\lambda_m\mu_m+\lambda_m\mu_e}
= \frac{\lambda_e}{\mu_m}.
\end{align}
Using (\ref{eq:globBal3}), we have
\begin{align}
\sum\limits_{N=1}^{\infty}\pi_{N/1} &= \frac{\mu_m}{\mu_e}\sum\limits_{N=1}^{\infty}\pi_{N/3} = \frac{\mu_m}{\mu_e}\frac{\lambda_e}{\mu_m} = \frac{\lambda_e}{\mu_e}.
\end{align}

\section{Proof of Proposition \ref{prop:avgGFcomp}}
\label{app:prop2Proof}
\begin{proof}
Recall that the SD architecture's average gate fidelity is given by
\begin{align}
F_{avg}^{(1)}(\mathcal{N}_t,G) &=  \int\limits_{0}^{\infty}F(\mathcal{N}_t,G)f_{W_1}(t)dt\nonumber\\
&= \int\limits_{0}^{\infty}F(\mathcal{N}_t,G)\left(\lambda_e e^{-\mu_e t} +\lambda_e e^{-\mu_m^{(1)}t} + \left(1-\frac{\lambda_e}{\mu_e}-\frac{\lambda_e}{\mu_m^{(1)}}\right)\delta(t)\right)dt\nonumber\\
&= \lambda_e\int\limits_{0}^{\infty}F(\mathcal{N}_t,G)\left(e^{-\mu_e t} +e^{-\mu_m^{(1)}t} \right)dt +1-\frac{\lambda_e}{\mu_e}-\frac{\lambda_e}{\mu_m^{(1)}},
\label{eq:FavgArch1}
\end{align}
while for the DD architecture,
\begin{align}
F_{avg}^{(2)} (\mathcal{N}_t,G) &=  \int\limits_{0}^{\infty}F(\mathcal{N}_t,G)f_{W_2}(t)dt\nonumber\\
&= \int\limits_{0}^{\infty}F(\mathcal{N}_t,G)\left(\lambda_e e^{-\mu_m^{(2)}t} + \left(1-\frac{\lambda_e}{\mu_m^{(2)}}\right)\delta(t)\right)dt\nonumber\\
&= \lambda_e \int\limits_{0}^{\infty}F(\mathcal{N}_t,G)e^{-\mu_m^{(2)}t}dt+ 1-\frac{\lambda_e}{\mu_m^{(2)}}.
\end{align}
Next, note that for a function $0\leq g(t)\leq 1$ and $x,y>0$ with $x\geq y$,
\begin{align}
\int\limits_{0}^{\infty}g(t) e^{-xt}dt -\frac{1}{x} \geq \int\limits_{0}^{\infty}g(t) e^{-yt}dt -\frac{1}{y},
\label{eq:missingbit}
\end{align}
since
\begin{align*}
\int\limits_{0}^{\infty}g(t) (e^{-yt}- e^{-xt})dt  \leq \int\limits_{0}^{\infty}(e^{-yt}- e^{-xt})dt = \frac{1}{y} - \frac{1}{x}.
\end{align*}
Since $\mu_m^{(2)}\geq \mu_e$, it follows from (\ref{eq:missingbit}) that 
\begin{align}
F_{avg}^{(2)} \geq \lambda_e \int\limits_{0}^{\infty}F(\mathcal{N}_t,G)e^{-\mu_et}dt+ 1-\frac{\lambda_e}{\mu_e}.
\label{eq:needtoprove}
\end{align}
Call the quantity on the right-hand of (\ref{eq:needtoprove}) $\hat{F}$. Note that
\begin{align}
\hat{F} - F_{avg}^{(1)}
&= \lambda_e \int\limits_{0}^{\infty}F(\mathcal{N}_t,G)e^{-\mu_et}dt+ 1-\frac{\lambda_e}{\mu_e} - \left(\lambda_e\int\limits_{0}^{\infty}F(\mathcal{N}_t,G)\left(e^{-\mu_e t} +e^{-\mu_m^{(1)}t} \right)dt +1-\frac{\lambda_e}{\mu_e}-\frac{\lambda_e}{\mu_m^{(1)}}\right)\nonumber\\
&= \frac{\lambda_e}{\mu_m^{(1)}}- \lambda_e\int\limits_{0}^{\infty}F(\mathcal{N}_t,G)e^{-\mu_m^{(1)}t}dt \geq 0,
\label{eq:fhat}
\end{align}
where the last inequality follows from the fact that $F(\mathcal{N}_t,G)\leq 1$, for any $t$. By (\ref{eq:fhat}) and (\ref{eq:needtoprove}), we conclude that $F_{avg}^{(2)}\geq F_{avg}^{(1)}$.
\end{proof}

\section{Proof of Bound (\ref{eq:compositeBound})}
\label{app:boundProof}
To prove (\ref{eq:compositeBound}), we first prove the following useful facts.
\begin{lemma}
\label{lemma:simpleFact}
For any reals $a$, $x$, and $y$ all greater than zero, if $x>y$, then
\begin{align}
\frac{x}{ax+1} > \frac{y}{ay+1}.
\label{eq:simpleFact}
\end{align}
\end{lemma}
\begin{proof}{(Lemma \ref{lemma:simpleFact})}
Assume for a contradiction that (\ref{eq:simpleFact}) is not true. Then, it must be that
\begin{align*}
\frac{x}{ax+1} &\leq \frac{y}{ay+1},\implies
x(ay+1) \leq y(ax+1),\implies
x \leq y,
\end{align*}
which contradicts the hypothesis.
\end{proof}
\begin{lemma}
\label{lemma:simpleFactTwo}
For any reals $a$, $b$, and $x$ all greater than zero,
\begin{align}
    \frac{1}{a(ax+1)} + \frac{1}{b(bx+1)} < \frac{2}{\frac{\sqrt{2}ab}{a+b}\left(\frac{\sqrt{2}ab}{a+b}x+1\right)}.
    \label{eq:lemma5Eq}
\end{align}
\end{lemma}
\begin{proof}{(Lemma \ref{lemma:simpleFactTwo})} Let $c=\sqrt{2}ab/(a+b)$. Then
(\ref{eq:lemma5Eq}) is equivalent to
\begin{align}
\frac{(a^2+b^2)x + a+ b}{ab(ax+1)(bx+1)} &< \frac{2}{c\left(cx+1\right)},\\
((a^2+b^2)x + a+b)c\left(cx+1\right) &< 2ab(ax+1)(bx+1),\\
(a^2+b^2)c^2x^2 + \left((a+b)c+
a^2+b^2\right)cx + (a+b)c
&< 2ab(abx^2 + (a+b)x +1)\label{eq:60}.
\end{align}
Let us verify (\ref{eq:60}) by comparing the coefficients of terms $x^j$, $j=2,1,0$, on each side. First, consider the coefficients of $x^2$: on the left-hand side of (\ref{eq:60}), we have
\begin{align}
    (a^2+b^2)c^2 = (a^2+b^2)\frac{2a^2b^2}{(a+b)^2},
\end{align}
and it is easy to see that this quantity is less than $2a^2b^2$, the coefficient of $x^2$ on the right side of (\ref{eq:60}). Similarly, for the coefficient of $x$ on the left hand side of (\ref{eq:60}), we have
\begin{align}
    \left((a+b)c+
a^2+b^2\right)c = \left(\sqrt{2}ab+
a^2+b^2\right)\frac{\sqrt{2}ab}{a+b},
\end{align}
which is less than $2ab(a+b)$, the coefficient of $x$ on the right side of (\ref{eq:60}). Finally, the constant term on the left side of (\ref{eq:60}) is $(a+b)c = \sqrt{2}ab$, while the constant term on the right is larger ($2ab$).
\end{proof}
\begin{proof}{(Proposition \ref{prop:compositeBound})}
For this proof, it is useful to keep in mind that $T_1$ is always greater than $T_2$ for any architecture. From this fact, and by Lemma \ref{lemma:simpleFact}, it follows that to satisfy (\ref{eq:compositeCond}), it suffices to ensure that
\begin{align}
\frac{1}{\mu_e} + \frac{1}{\mu_m^{(1)}}
- \frac{T_2^{(1)}}{\mu_e T_2^{(1)} + 1} - \frac{T_2^{(1)}}{\mu^{(1)}_m T_2^{(1)} + 1}
<\frac{1}{\mu_m^{(2)}}- \frac{T_1^{(2)}}{\mu^{(2)}_m T_1^{(2)} + 1},
\label{eq:compositeBound1}
\end{align}
which we obtained by substituting each $T_1^{(1)}$ on the left-hand side of (\ref{eq:compositeCond}) with $T_2^{(1)}$ and each $T_2^{(2)}$ with $T_1^{(2)}$ on the right-hand side of (\ref{eq:compositeCond}). Next, (\ref{eq:compositeBound1}) is equivalent to
\begin{align}
\frac{1}{\mu_e(\mu_e T_2^{(1)} + 1)} + \frac{1}{\mu_m^{(1)}(\mu^{(1)}_m T_2^{(1)} + 1)}
<\frac{1}{\mu_m^{(2)}(\mu^{(2)}_m T_1^{(2)} + 1)}.
\label{eq:prop2Ineq}
\end{align}
To ensure (\ref{eq:prop2Ineq}) holds, it is sufficient to have 
\begin{align}
    \frac{2}{\frac{\sqrt{2}\mu_e \mu_m^{(1)}}{\mu_e+\mu_m^{(1)}}\left(\frac{\sqrt{2}\mu_e \mu_m^{(1)}}{\mu_e+\mu_m^{(1)}}T_2^{(1)}+1\right)}
    <\frac{1}{\mu_m^{(2)}(\mu^{(2)}_m T_1^{(2)} + 1)}.
    \label{eq:65}
\end{align}
Solving (\ref{eq:65}) for $T_2^{(1)}$ yields (\ref{eq:compositeBound}).
\end{proof}

\section{Nitrogen-Vacancy Center in Diamond}
\label{nv_in_diamond_app}
\subsection{Overview}
\label{app:NVoverview}
Here, we provide a brief overview of the Nitrogen-Vacancy (NV) center in diamond platform and its characteristics as relevant to our problem. For additional information, we refer the reader to, for example, \cite{linkLayerPaper}. The NV center in diamond platform is a few-qubit (at present, no more than 10 \cite{bradley2019ten}) quantum processor capable of executing arbitrary quantum gates and measurements and has demonstrated entanglement establishment over a distance of $1.3$km \cite{Hensen_2015}.  This hardware has also been used to demonstrate key quantum network protocols required for long-distance networking such as entanglement swapping and distillation \cite{pompili2021realization,kalb2017entanglement}. Qubits in the NV center in diamond platform may be divided into two types, \emph{communication qubits} and \emph{storage qubits}.  Communication qubits are those that are equipped with optical interfaces, allowing them to establish entanglement with communication qubits in other quantum processors while storage qubits, on the other hand, are only capable of storing quantum states and having quantum gates applied to them. In the particular case of the NV center in diamond, the set of quantum gates that can be applied to storage qubits is limited to a time-dependent rotation about the $Z$ axis, while arbitrary quantum gates may be applied to the communication qubit.

Interactions between qubits in the NV center in diamond platform are mediated by an electronic spin (the NV) that acts as a communication qubit.  The remaining qubits are C$^{13}$ spins that act as storage qubits which are magnetically coupled to the electronic spin.  This property of NV in diamond hardware restricts the parallelism of quantum gates on different qubits as multiple quantum gates may not be applied to the NV at the same time, resulting in serial execution of quantum gates.  Furthermore, most quantum gates may only be performed on quantum states held by the NV qubit, while C$^{13}$ spins may only be initialized and undergo $Z$-rotations.  This means that the application of a quantum gate to a state held by a storage qubit requires exchanging the quantum state of the NV with said storage qubit.

One implication of these physical restrictions on computation with the NV center in diamond is that
for our model of the single-NV architecture (Section \ref{sec:model}),
computational jobs may require state transfers before they can be serviced. Specifically, such a situation may arise whenever the system is awaiting a state transfer request (to move the state of a newly-entangled qubit from the NV to a storage qubit), while one or more computation requests are in the queue. Recall that during such ``idle'' periods, according to our modeling framework computation is allowed in the SD architecture. However, depending on the type of computation that is being requested (see discussion above), the processor may need to perform a state transfer prior to the computation (NV $\to$ storage qubit) to free up the NV qubit, perform the computation (recall that when $\mu_c=\infty$, all computational jobs may be processed instantaneously), and finally, perform another state transfer (storage qubit $\to$ NV) to move the state of the entangled qubit back to the NV. One may wonder: why perform the latter move when the state of the entangled qubit must eventually be moved to storage? The reason is that future computation requests may require the use of the entangled qubit in the NV. On the other hand, it may be possible that all computation requests involve the entangled qubit in the NV, in which case the extra transfers are not required. Since we have no knowledge of the computational request requirements a priori, we simply assume within our model that ``implicit'' state transfers are performed when necessary, and that they also (akin to computation requests) consume a negligible amount of time.
In a manner, our simulations in Section~\ref{sec:SimNumerObs} relax this assumption by varying computation request processing rates.  Applications that infrequently shuffle quantum states between qubits correspond to the case where $\mu_c$ is very large while applications that move quantum states more frequently correspond to the case when $\mu_c$ is smaller and comparable to $\mu_m$.

In addition to the limitations on quantum gates, the NV qubit is the sole communication qubit that may be equipped with an optical interface for establishing entanglement with qubits in other quantum processors.  This further restricts parallelism of operations as quantum gates may not be applied to any qubits when the processor is being used to establish entanglement.  Studies on experimental realizations of such hardware have shown that the fidelity and the rate at which entangled states may be established between such processors decrease as the distance over which entanglement must be established increases \cite{linkLayerPaper,rozpkedek2019near}. Combined with existing challenges in emitted photon collection \cite{ruf2021quantum}, this means that establishing entanglement will dominate the execution time of applications and incur additional latency for performing local quantum gates between qubits. Furthermore, the process of establishing entanglement introduces noise on quantum states that are stored in the remaining qubits of the system \cite{kalb2018dephasing}, which can severely reduce the fidelity of stored states.

\subsection{Transferring Quantum States}
\label{sec:nv_state_transfer}
\subsubsection{Single-Device NV}
Our NetSquid implementation for simulating the transfer of a quantum state from the NV electron spin qubit into carbon storage qubit is performed through a sequence of gates shown in the circuit of Figure \ref{fig:arch1_move} \cite{kalb2017entanglement}.

\begin{figure}[ht!]
    \centering
    \includegraphics[width=0.75\linewidth]{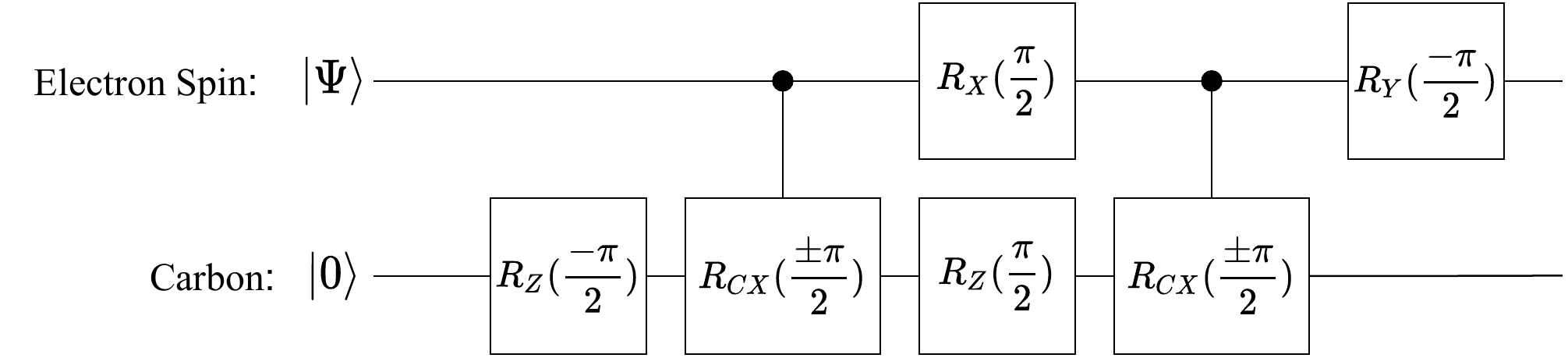}
    \caption{Quantum circuit for moving a quantum state from the NV into carbon storage. Each gate applies noise onto the qubits it interacts with.}
    \label{fig:arch1_move}
\end{figure}

Here, the electron spin is in some quantum state $|\Psi\rangle$ and the carbon storage is initialized to the state $|0\rangle$. Gates are applied to each qubit in order from left to right and lines from the NV to a gate on the carbon denote a controlled quantum gate. For reference, the quantum gates used in our simulations are defined as
\begin{align}
R_X(\theta) &= \begin{bmatrix} \cos(\frac{\theta}{2}) & -i\sin(\frac{\theta}{2})\\ -i\sin(\frac{\theta}{2}) & \cos(\frac{\theta}{2}) \end{bmatrix}, \quad R_Y(\theta) = \begin{bmatrix} \cos(\frac{\theta}{2}) & -\sin(\frac{\theta}{2})\\ \sin(\frac{\theta}{2}) & \cos(\frac{\theta}{2}) \end{bmatrix}, \quad R_Z(\theta) = \begin{bmatrix} e^{-i\frac{\theta}{2}} 0\\ 0 & e^{i\frac{\theta}{2}} \end{bmatrix}, \nonumber \\
R_{CX}(\pm \theta) &= |0\rangle \langle0| \otimes R_X(\theta) + |1\rangle \langle 1| \otimes R_X(-\theta) = \begin{bmatrix} \cos(\frac{\theta}{2}) & -i\sin(\frac{\theta}{2}) & 0 & 0\\ -i\sin(\frac{\theta}{2}) & \cos(\frac{\theta}{2}) & 0 & 0\\ 0 & 0 & \cos(\frac{\theta}{2}) & i\sin(\frac{\theta}{2})\\ 0 & 0 & i\sin(\frac{\theta}{2}) & \cos(\frac{\theta}{2}) \\ 
\end{bmatrix}, \nonumber \\
R_{CY}(\pm \theta) &= |0\rangle \langle0| \otimes R_Y(\theta) + |1\rangle \langle 1| \otimes R_Y(-\theta) = \begin{bmatrix} \cos(\frac{\theta}{2}) & -\sin(\frac{\theta}{2}) & 0 & 0\\ \sin(\frac{\theta}{2}) & \cos(\frac{\theta}{2}) & 0 & 0\\ 0 & 0 & \cos(\frac{\theta}{2}) & \sin(\frac{\theta}{2})\\ 0 & 0 & -\sin(\frac{\theta}{2}) & \cos(\frac{\theta}{2}) \\ 
\end{bmatrix}, \nonumber
\end{align}
where $\otimes$ denotes the tensor (Kronecker) product of two matrices.

In NetSquid, gates in the NV center in diamond platform are modeled as the application of the perfect gate $G$ after applying time-independent depolarizing noise $\mathcal{D}_G$ that is parameterized by an associated depolarizing probability $p_G$ depending on the gate G being applied. For the gate sequence in Figure \ref{fig:arch1_move}, the depolarizing parameters for each gate are summarized in Table \ref{tab:depol_params}.
\begin{table}[ht!]
\caption{Depolarizing parameters for gates in the NV hardware. We remark that no two NV devices are exactly identical. Individual values have not been realized simultaneously for producing entanglement which would allow a direct comparison to simulation.  We thus focus on simulation parameters that enable a comparison to entanglement generation hardware and provide references to motivations for our chosen values.}
\begin{tabular}{|l|l|l|l|}
\hline
 & Qubits Depolarized & $p_G$ NV & $p_G$ Carbon \\ \hline
Electron Initialization \cite{Reiserer_2016} & Electron & 0.02 & - \\ \hline
Electron $R_X(\theta)$ \cite{kalb2017entanglement} & - & - & - \\ \hline
Electron $R_Y(\theta)$ \cite{kalb2017entanglement} & - & - & - \\ \hline
Carbon Initialization \cite{bradley2019ten} & Carbon & - & $0.006 / 4$ \\ \hline
Carbon $R_Z(\theta)$ \cite{taminiau2014universal} & Carbon & - & $0.001 / 3$ \\ \hline
Electron-Carbon $R_{CX}(\pm\theta)$ \cite{kalb2017entanglement} & Electron and Carbon & 0.005 & 0.005 \\ \hline
Electron-Carbon $R_{CY}(\pm\theta)$ \cite{kalb2017entanglement} & Electron and Carbon & 0.005 & 0.005 \\ \hline
\end{tabular}
\label{tab:depol_params}
\end{table}

As an example, consider applying the $R_Z(\frac{\pi}{2})$ to some state $\rho = |\psi\rangle\langle\psi|$ in a carbon qubit.  The new state is
\begin{align}
 \mathcal{D}_G \left(R_Z\left(\frac{\pi}{2}\right)\rho R_Z^{\dagger}\left(\frac{\pi}{2}\right)\right)  \nonumber
\end{align}
using depolarizing probability $p_{R_Z}=\frac{0.001}{3}$ for $\mathcal{D}_G$.

\subsubsection{Double-Device NV}
Our NetSquid implementation for simulating quantum state transfer from the NV electron spin of the networking device to the NV electron spin of the computing device expands upon the gate sequence in Figure \ref{fig:arch1_move} and makes use of quantum teleportation \cite{Bennett_1996}.  Once entanglement that was originally held by the networking device electron spin has been moved to the networking device carbon storage (using the gate sequence from Figure~\ref{fig:arch1_move}), the following gate sequence Figure \ref{fig:arch2_move} is performed in order to transfer the state to the computing device electron spin \cite{pompili2021realization}.

\begin{figure}[ht!]
    \centering
    \includegraphics[width=\linewidth]{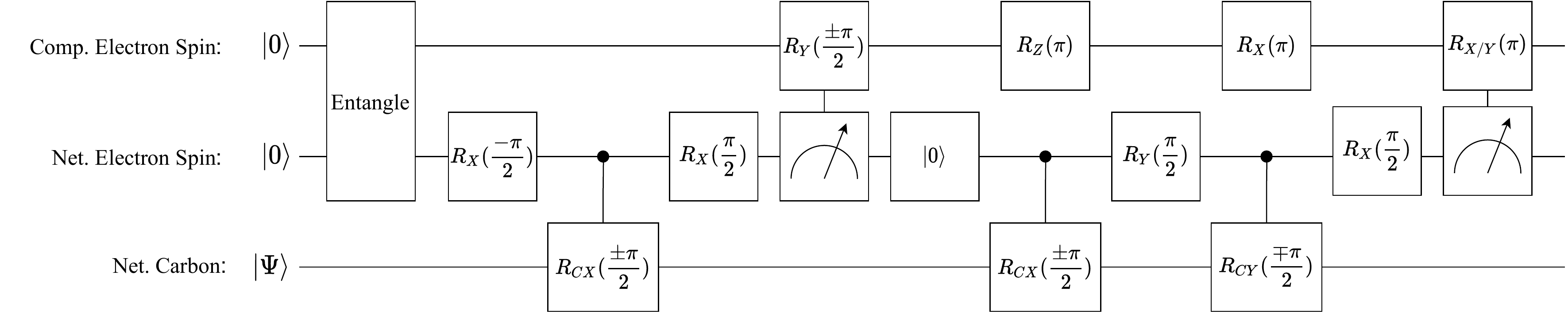}
    \caption{Quantum circuit for moving a quantum state from the networking device electron spin into the computing device electron spin.}
    \label{fig:arch2_move}
\end{figure}

Here, the "Entangle" box represents establishing entanglement between the electron spins of the networking device and the computing device, the $|0\rangle$ box represents initializing the networking device electron spin to the $|0\rangle$ state, and boxes containing arcs with an arrow represent measuring the qubit in the standard ($Z$) basis. The "Entangle" operation instantly establishes a perfect entangled state between the electron spins of the computing device and the networking device, thus giving an optimistic evaluation of the state transfer process in the DD architecture. Connections between measurement boxes and a gate on the computing electron spin represent conditional execution of the gate based on the measurement result. In the first case, measuring $|0\rangle$ on the networking electron spin means we perform a $R_Y(\pi)$ gate on the computing NV while measuring $|1\rangle$ on the networking NV means we perform a $R_Y(-\pi)$ gate.  Similarly for the second instance, here we perform a $R_X(\pi)$ gate on the computing electron spin if we measure $|0\rangle$ on the networking NV and a $R_Y(\pi)$ gate if we measure $|1\rangle$ on the networking electron spin.

\section{Fidelity Derivations}
\label{fidelity_derivation_app}
\subsection{Alternative Average Gate Fidelity Derivation for Noise Channels}
\label{app:avg_fidelity_derivation}
For completeness, we present an alternative, self-contained derivation of the average gate fidelity that depends only on a number of well-known tricks in quantum information. Our proof is based on the Choi-Jamiolkowski theorem establishing a duality between quantum states and quantum channels.
First, we will make use of the notion of the Choi state of a quantum channel.  Here, it will be sufficient to note that for channels $\mathcal{N}:\mathcal{S}_A \rightarrow \mathcal{S}_A$ where $\mathcal{S}_A$ denotes the state of a quantum system $A$, the Choi state for $\mathcal{N}$ is defined (see \eg,~\cite{wolf2012quantum}) as
\begin{align}\label{eq:ChoiState}
\tau_{\mathcal{N}} = \mathcal{N} \otimes \id_{A'} \left(\proj{\Phi_{AA'}}\right)\ ,
\end{align}
where 
\begin{align}
\ket{\Phi_{AA'}} = \frac{1}{\sqrt{d_A}}\sum_{j=1}^{d_A} \ket{j}_A \ket{j}_{A'}
\end{align}
is the maximally entangled state on $A$ and an identical quantum system $A'$. 
We will also make use of the partial transpose operation $\Gamma$. For any operator $M_{AA'} \in \Complex^{d_A\times d_A} \otimes \Complex^{d_A \times d_A}$ 
the partial transpose is defined as
\begin{align}
M^\Gamma_{AA'} &= \left(\sum_{ijk\ell} c_{ij}^{k\ell} \ket{i}\bra{j}_A \otimes \ket{k}\bra{\ell}_{A'}\right)^\Gamma =\sum_{ijk\ell}c_{ij}^{k\ell} \ket{i}\bra{j}_A \otimes \left(\ket{k}\bra{\ell}_{A'}\right)^T\ , 
\end{align}
which corresponds to taking the transpose on $A'$, but not on $A$. 
 
We first establish the following lemma.
\begin{lemma}
\label{lem:avgFid}
Let $\mathcal{N}_t:\mathcal{S}_A \rightarrow \mathcal{S}_A$ denote a noisy quantum channel depending on time $t$. For any quantum gate $G$ on system $A$, we have
\begin{align}
F(\mathcal{N}_t,G) = \frac{d_A}{d_{\rm sym}} \Tr\left[\Pi_{\rm sym} \tau_{\mathcal{N}_t}^\Gamma\right]\ , 
\end{align}
where $d_A$ is the dimension of the quantum system $A$, $\Pi_{\rm sym}$ is the projector onto the symmetric subspace of $\Complex^{d_A\times d_A} \otimes \Complex^{d_A \times d_A}$, $d_{\rm sym}$ is the dimension of said symmetric subspace, and $\tau_{\mathcal{N}_t}$ is the Choi state of $\mathcal{N}_t$.
\end{lemma}
\begin{proof}
Using the Choi-Jamiolkowski isomorphism (see \eg, \cite{wolf2012quantum}), we can rewrite
\begin{align}
\bra{\Psi} G^\dagger \mathcal{N}_t(G \proj{\Psi} G^\dagger) G \ket{\Psi} &= \Tr\left[G\ket{\Psi}\bra{\Psi}G^\dagger \mathcal{N}_t\left(G\proj{\Psi}G^\dagger\right)\right]\\
&= d_A \Tr\left[G\ket{\Psi}\bra{\Psi}G^\dagger \otimes \left(G \proj{\Psi}G^\dagger\right)^T  \tau_{\mathcal{N}_t}\right]\\
&= d_A \Tr\left[\left(G\ket{\Psi}\bra{\Psi}G^\dagger\right)^{\otimes 2} \tau_{\mathcal{N}_t}^\Gamma\right]\ .\label{eq:afterChoi}
\end{align}
Using the definition for the average gate fidelity~\eqref{eq:GF} and~\eqref{eq:afterChoi}, we can then write 
\begin{align}
F(\mathcal{N}_t,G) &= d_A \Tr\left[\left(\int d\Psi \left(G\proj{\Psi}G^\dagger\right)^{\otimes 2}\right) \tau_{\mathcal{N}_t}^\Gamma\right]\\
&=\frac{d_A}{d_{\rm sym}}\Tr\left[\Pi_{\rm sym} \tau_{\mathcal{N}_t}^\Gamma\right]\ , 
\end{align}
where we have used the fact that
\begin{align}
\int d\Psi \left(G\proj{\Psi}G^\dagger\right)^{\otimes 2} &=
\int d\Psi \proj{\Psi}^{\otimes 2}\\
&=\frac{\Pi_{\rm sym}}{d_{\rm sym}}\ .\label{eq:linkEquation}
\end{align}
\end{proof}
Given the little lemma above, we can now readily evaluate the gate fidelity for any channel of interest in two steps: first, we need an expression for the projector $\Pi_{\rm sym}$ onto the symmetric subspace. It is well known (see \eg,~\cite{renes2004symmetric}) that the 
a full set of so-called mutually unbiased basis in dimension $d=2^n$ forms a $2$-design, \ie,
\begin{align}
\int d\Psi \proj{\Psi}^{\otimes 2} &= \frac{1}{d(d+1)} \sum_{\theta}\sum_j \proj{j_\theta}\ ,\\
&=\frac{2\Pi_{\rm sym}}{d(d+1)},
\end{align}
where the sum extends over bases indexed by $\theta$ and $\ket{j_\theta}$ denotes the $j$-th basis state of basis $\theta$. For $d=2$, \ie, a single qubit, 
these bases are simply the eigenbases of the operators $X$, $Z$ and $Y$ defined in Section~\ref{sec:quantum_bg}. Using~\eqref{eq:linkEquation}, this allows us to write
\begin{align}
\Pi_{\rm sym} = \frac{1}{2}\sum_{\theta \in \{X,Z,Y\}} \sum_j \proj{j_\theta}\ ,
\end{align}
where $d_{\rm sym} = 3$. In such small dimensions, it is also easy to write
\begin{align}
\Pi_{\rm sym} = \proj{\Phi_{00}} + \proj{\Phi_{01}} + \proj{\Phi_{10}}\ ,
\end{align}
where $\ket{\Phi_{ab}} = \id \otimes X^a Z^b \ket{\Phi}$ denotes the first three Bell states (excluding the singlet). Second, we need to compute the Choi states of the noise channels defined in Section~\ref{sec:noise}, which can readily be achieved by using~\eqref{eq:ChoiState}. 

\if{false}
\begin{align}
\mathcal{A}(\rho) &= M_0 \rho M_0^{\dagger} + M_1 \rho M_1^{\dagger},
\label{eq:damping}
\end{align}
where $M_0,M_1$ have the form
\begin{align}
M_0 = \begin{bmatrix} 1 & 0\\ 0 & \sqrt{1 - \gamma} \end{bmatrix}, 
M_1 = \begin{bmatrix} 0 & \sqrt{\gamma}\\ 0 & 0 \end{bmatrix},
\end{align}
where $\gamma = (1-e^{-\frac{t}{T_1}})$ for a fixed $T_1$ characterizing the effects of the amplitude damping channel. The composite noise channel $\mathcal{C}$ that models the noise experienced by qubits stored in the NV hardware is 
\begin{align}
\mathcal{C}(\rho) &= \mathcal{P}(\mathcal{A}(\rho)) = (1-p)\mathcal{A}(\rho) + pZ\mathcal{A}(\rho)Z \\
&= (1-p)(M_0 \rho M_0^{\dagger} + M_1 \rho M_1^{\dagger}) + pZ(M_0 \rho M_0^{\dagger} + M_1 \rho M_1^{\dagger})Z,
\end{align}
\subsection{Calculating the average fidelity for specific noise channels}
\label{sec:specificNoise}
With the help of Lemma~\ref{lem:avgFid}, we are now ready to evaluate the average gate fidelity for common forms of quantum noise. Note that this average fidelity depends only on the quantum gate $G$ being unitary, but not on the exact choice of gate. This means that also for a quantum memory where $G=\id$ we obtain the average fidelity of storing a quantum state due to suspending operations during entangling operations. M., P., and R. Horodecki have presented a beautiful formula \cite{horodecki1999general} connecting the average gate fidelity $F_{avg}(\mathcal{N}_t, G)$ to the average entanglement fidelity $F_e(\mathcal{N}_t)$ by 
\begin{align}
F_{avg}(\mathcal{N}_t) &= \frac{dF_e(\mathcal{N}_t) + 1}{d + 1}
\label{eq:ent_fid}
\end{align}
Thus using $G=\id$ allows us to calculate the corresponding entanglement fidelity for storing a quantum state.
A simple calculation shows that for depolarizing noise with $p = \frac{1}{4}\left(1 - e^{-t/T}\right)$ as defined in~\eqref{eq:depol} we have for any $G$
\begin{align}
F(\mathcal{D}, G) = \frac{1}{2}\left(1 + e^{-\frac{t}{T}}\right)\ .
\end{align}
Similarly, for dephasing noise with $p=\left(1-e^{-t/T_2}\right)$
\begin{align}
F(\mathcal{P}, G) = \frac{1}{3}\left(1 + 2e^{-\frac{t}{T_2}}\right)\ .
\end{align}
For the NV noise model, using $\gamma = (1-e^{-\frac{t}{T_1}})$ for $\mathcal{A}$ and $p=\left(1-e^{-t/T_2}\right)$ for $\mathcal{D}$, we obtain
\begin{align}
F(\mathcal{C}, G) = \frac{1}{6}\left(3 + e^{-\frac{t}{T_1}} + 2e^{-\frac{t}{T_2}}\right).
\end{align}
\fi

\subsection{Average Post-Move Entanglement Fidelity in Single-NV Architecture}
\label{app:postMoveEntFidSingleNV}
Using Lemma~\ref{lem:avgFid} we may also obtain analytic expressions for the entanglement fidelity of the state that was moved into memory.  

We may analytically compute the gate fidelity for the move to memory gate sequence in Figure \ref{fig:arch1_move} as
\begin{align}
\begin{split}
F^{pm} &= \frac{1}{6}(3 - 6p_{R_X} + (p_{init}-1)(p_{R_Z}-1)^2(p_{R_{CX}}-1)^2(2p_{R_X}-1)e^{-\frac{t}{T_1}} \\
&+ (2 + p_{init}(p_{R_Z}-1)-p_{R_Z})(p_{R_Z}-1)(p_{R_{CX}}-1)^3(4p_{R_X}-1)e^{-\frac{t}{T_2}}),
\end{split}
\end{align}
\noindent where $p_{init}$, $p_{R_Z}$, $p_{R_X}$ and $p_{R_{CX}}$ are the depolarizing probabilities for Carbon Initialization, Carbon $R_Z$ rotations, Electron $R_X$ rotations, and Electron-Carbon $R_{CX}$ respectively.  We may now obtain the average post-move gate fidelity by integrating over the waiting time distribution of move requests,
\begin{align}
\begin{split}
F_{avg}^{pm} &= \int_0^{\infty} \lambda_m e^{-\lambda_m t} \frac{1}{6}(3 - 6p_{R_X} + (p_{init}-1)(p_{R_Z}-1)^2(p_{R_{CX}}-1)^2(2p_{R_X}-1)e^{-\frac{t}{T_1}} \\
&+ (2 + p_{init}(p_{R_Z}-1)-p_{R_Z})(p_{R_Z}-1)(p_{R_{CX}}-1)^3(4p_{R_X}-1)e^{-\frac{t}{T_2}})dt
\end{split}\\
\begin{split}
&= \frac{1}{2} - p_{R_X} + \frac{1}{6}(p_{init}-1)(p_{R_Z}-1)^2(p_{R_{CX}}-1)^2(2p_{R_X}-1)\frac{\lambda_m T_1}{\lambda_m T_1 + 1}\\
&+ \frac{1}{6}(2 + p_{init}(p_{R_Z}-1)-p_{R_Z})(p_{R_Z}-1)(p_{R_{CX}}-1)^3(4p_{R_X}-1) \frac{\lambda_m T_2}{\lambda_m T_2 + 1}.
\end{split}
\end{align}
Using $p_{init}=0.006/4,p_{R_Z}=0.001/3,p_{R_X}=0$ and $p_{R_{CX}}=0.005$, we obtain a post-move entanglement fidelity of
\begin{align}
F_e^{pm} = \frac{3F^{pm} - 1}{2} &= \frac{3(0.5 + 0.158682e^{-\frac{t}{T_1}} + 0.312165e^{-\frac{t}{T_2}}) - 1}{2}\\
&= 0.25 + 0.238023e^{-\frac{t}{T_1}} + 0.4682475e^{-\frac{t}{T_2}},
\end{align}
and an average post-move entanglement fidelity of
\begin{align}
F_{e,avg}^{pm} = \frac{3F_{avg}^{pm} - 1}{2} &= \frac{3(0.5 + 0.158682\frac{\lambda_m T_1}{\lambda_m T_1 + 1} + 0.312165\frac{\lambda_m T_2}{\lambda_m T_2 + 1}) - 1}{2}\\
&= 0.25 + 0.238023\frac{\lambda_m T_1}{\lambda_m T_1 + 1} + 0.4682475\frac{\lambda_m T_2}{\lambda_m T_2 + 1}.
\label{eq:arch1_pm_fidelity}
\end{align}
From (\ref{eq:arch1_pm_fidelity}) we see that we have an upper bound of $\approx 0.956$ on the average post-move fidelity.

\section{Single-Device and Double-Device Average Gate Fidelity}
\label{app:indivFidelityPlots}
\begin{figure}[t]
\centering
\subfloat[single-device]{\includegraphics[width=0.4\textwidth]{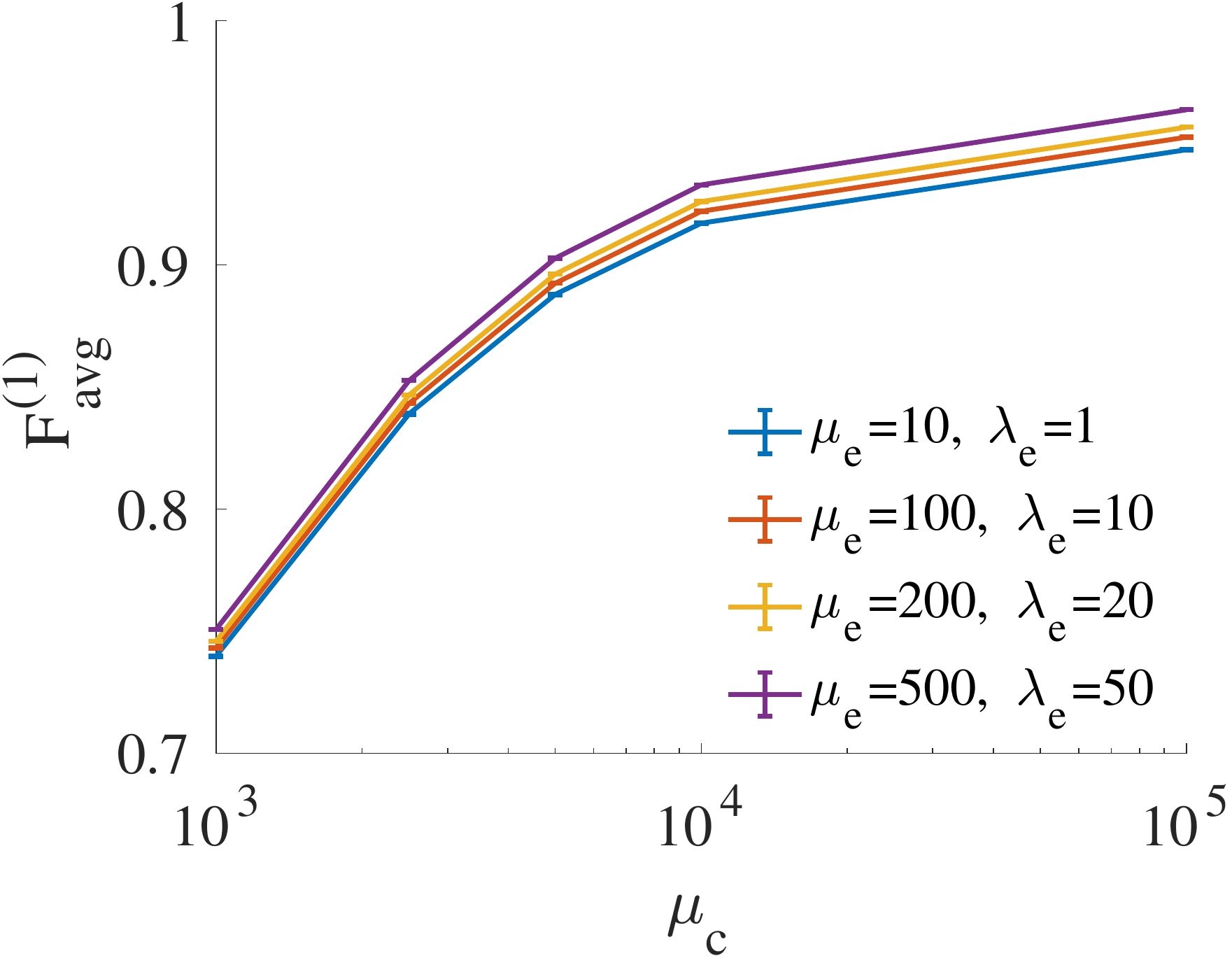}\label{fig:sep_sim_comp_fidelity_low_T1T2_a}}\qquad
\subfloat[double-device]{\includegraphics[width=0.4\textwidth]{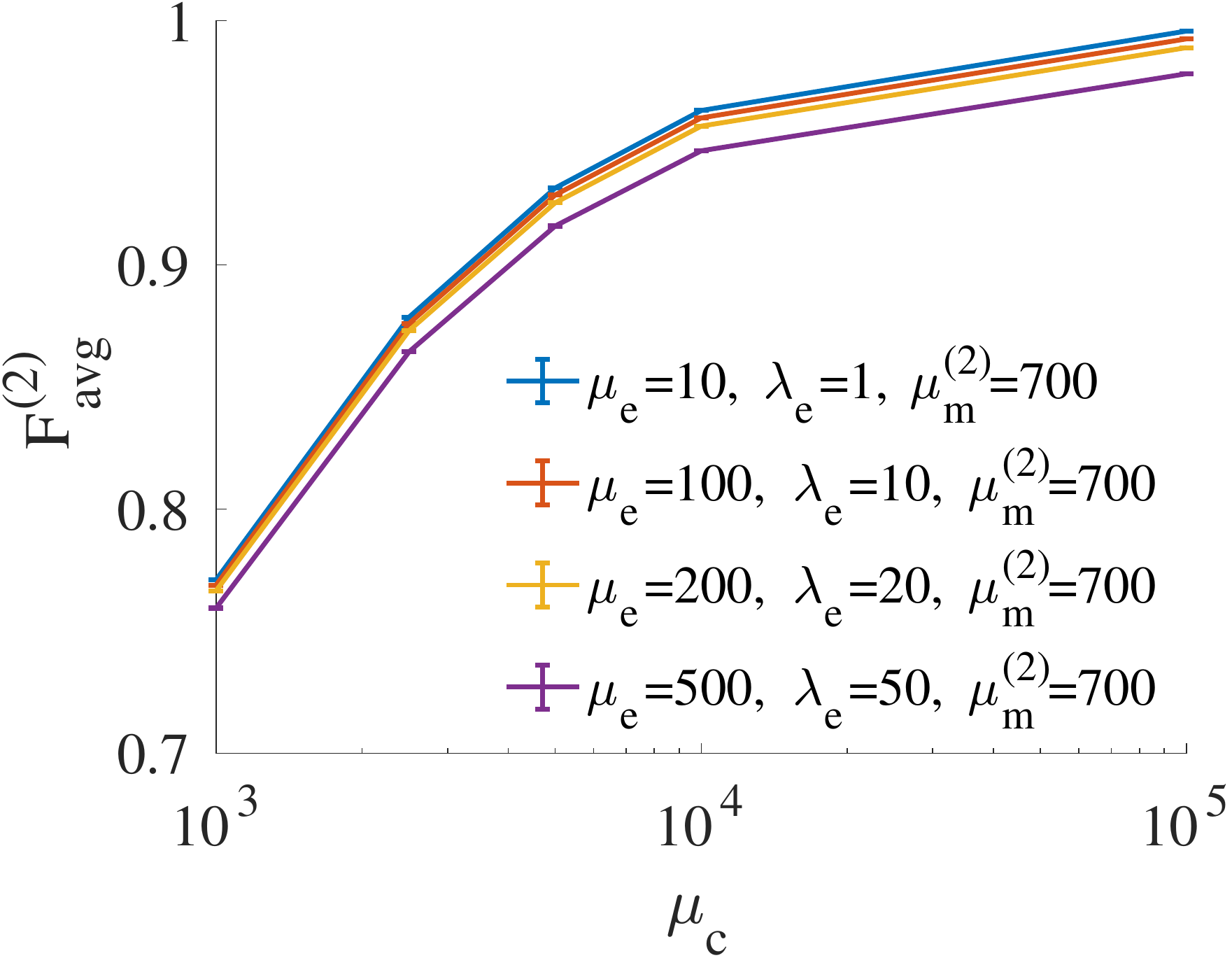}\label{fig:sep_sim_comp_fidelity_low_T1T2_b}}
\caption{Average gate fidelities for computational jobs for the SD architecture, $(a)$, and the DD architecture, $(b)$. Here, $T_1^{(1)}=T_1^{(2)}=0.00286$s and $T_2^{(1)}=T_2^{(2)}=0.001$s; $\mu_m^{(1)}=1667$Hz for all SD simulations. For all experiments, $\lambda_c = 150$ and $\lambda_m = 1000$.}
\label{fig:sep_sim_comp_fidelity_low_T1T2}
\end{figure}
\begin{figure}[t]
\centering
\subfloat[single-device]{\includegraphics[width=0.4\textwidth]{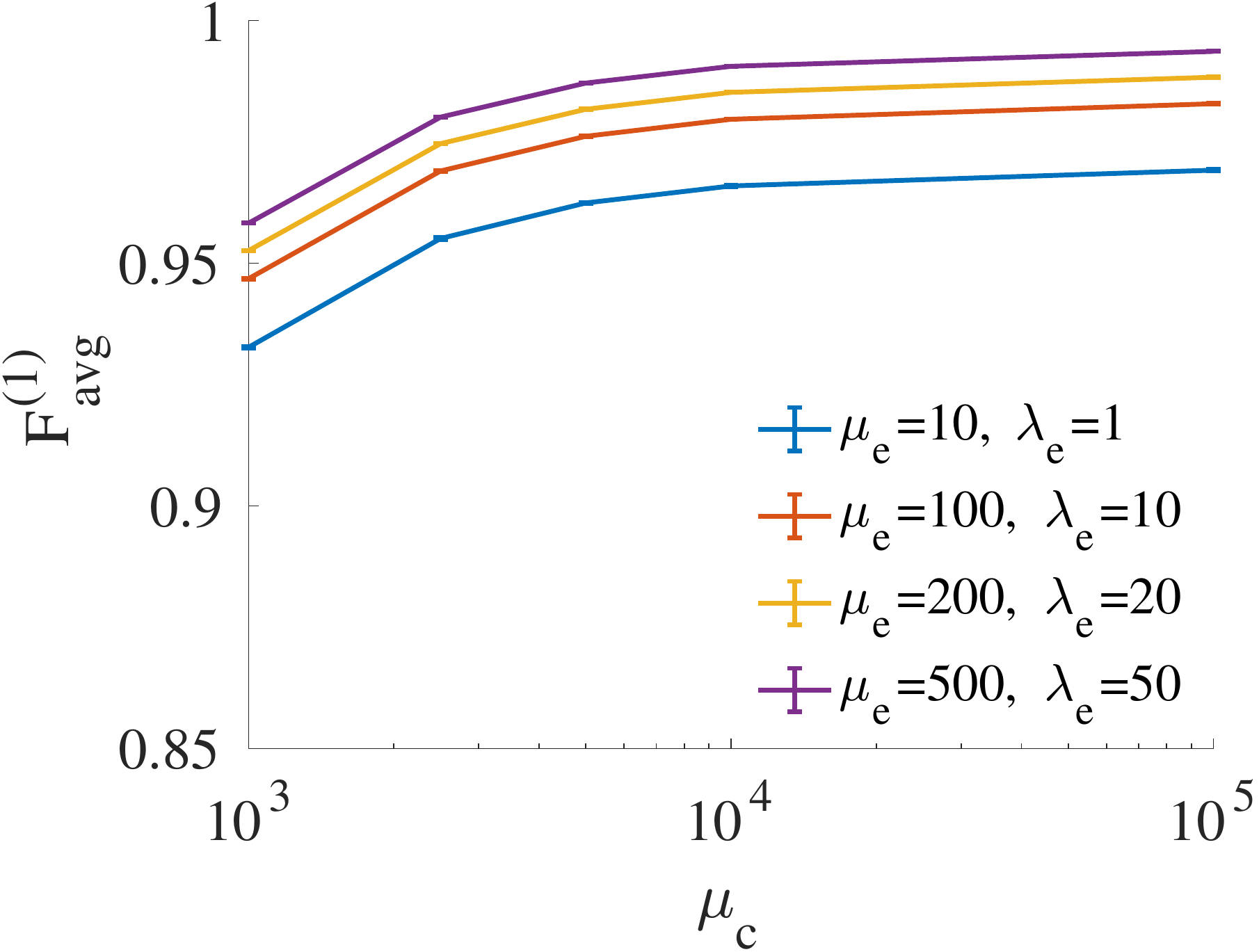}}\qquad
\subfloat[double-device]{\includegraphics[width=0.4\textwidth]{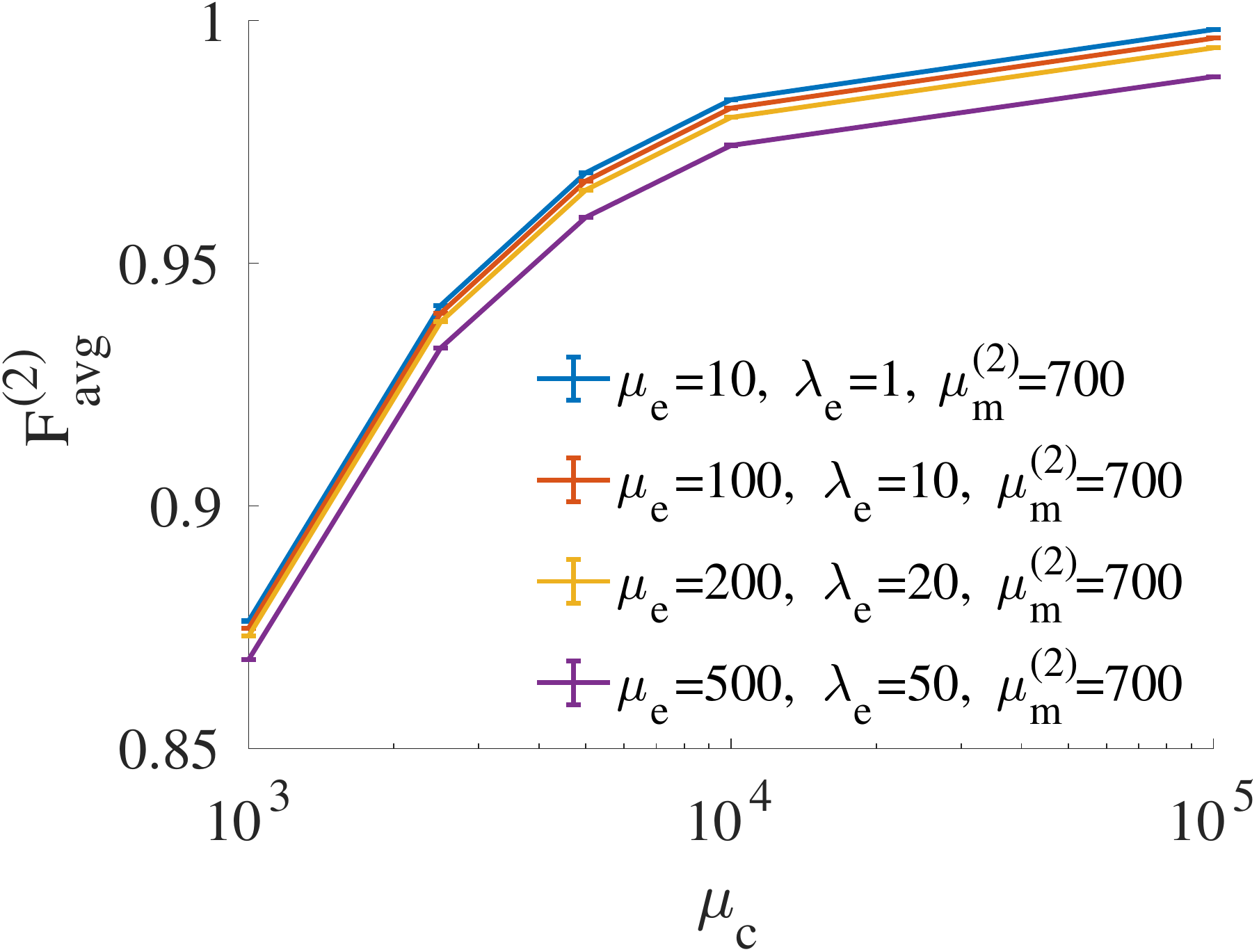}}
\caption{Average gate fidelities for computational jobs for the SD architecture, $(a)$, and the DD architecture, $(b)$. Here, $T_1^{(1)}=10$s, $T_2^{(1)}=0.01$s, $T_1^{(2)}=2$s, and $T_2^{(2)}=0.002$s; $\mu_m^{(1)}=1667$Hz for all SD simulations. For all experiments, $\lambda_c = 150$ and $\lambda_m = 1000$.}
\label{fig:sep_sim_comp_fidelity_high_T1T2}
\end{figure}
In Figures \ref{fig:sim_comp_fidelity_low_T1T2} and \ref{fig:sim_comp_fidelity_high_T1T2}, we presented differences between  the SD and DD architecture average gate fidelities. In Figures \ref{fig:sep_sim_comp_fidelity_low_T1T2} and \ref{fig:sep_sim_comp_fidelity_high_T1T2}, we observe the individual average gate fidelity of each architecture design for various entanglement request and generation scenarios, as the computational job processing rate $\mu_c$ varies. Figure \ref{fig:sep_sim_comp_fidelity_low_T1T2} corresponds to the same memory quality regime as Figure \ref{fig:sim_comp_fidelity_low_T1T2}: namely, one in which the memory lifetimes are equal for the two architectures.
Figure \ref{fig:sep_sim_comp_fidelity_high_T1T2} corresponds to the same memory quality regime as Figure \ref{fig:sim_comp_fidelity_high_T1T2}, where the memory lifetimes for the DD architecture are five times shorter than that of the SD architecture.

In both figures, we observe that as the entanglement generation rate increases, the average gate fidelity of the SD architecture increases, while the fidelity decreases for the DD architecture (in these particular examples, this holds even as the entanglement request rate scales up with the generation rate, although it may not hold in general for higher request rates). This is an expected result: recall that in the SD architecture, computational jobs wait both for entanglement generation as well as for moving requests, while in the DD architecture they only wait for moving requests. In addition, moving jobs have non-preemptive priority over computational jobs when $\mu_c <\infty$. Thus, for a fixed moving request rate, faster entanglement generation only aids computational jobs in the SD design. In contrast, in the DD design, higher entanglement generation rates lead to more moving requests, thus increasing the likelihood of computation being interrupted by these requests.

\end{document}